\documentclass[12pt,english]{article}
\usepackage{a4}
\usepackage[round]{natbib}
\usepackage{babel}

\usepackage[margin=1.2in]{geometry}
\usepackage{array}
\usepackage{amsmath,amssymb,bm}
\usepackage{amsthm}
\usepackage{mathtools}
\usepackage{amsbsy}
\usepackage[colorlinks = true,linkcolor = purple,urlcolor  = blue,citecolor = blue,anchorcolor = blue]{hyperref}
\usepackage{setspace}
\usepackage{xfrac}
\usepackage{tikz}
\usetikzlibrary{positioning}
\usepackage[affil-it]{authblk}

\usepackage{mathptmx} %Times New Roman

\begin{document}

\setlength{\extrarowheight}{3pt} % table row height

\title{Group Knowledge and Individual Introspection\footnote{Many thanks to Hans van Ditmarsch, Klaus Kultti and Hannu Vartiainen for helpful conversations, and to anonymous reviewers from the Workshop on Epistemic Game Theory (Maastricht, July 2022) and from the XIX TARK Conference (Oxford, June 2023) for their helpful comments. I am especially grateful to an anonymous reviewer from Theory and Decision for detailed and insightful comments that helped me greatly to improve the paper. An abstract of this paper was published in the Proceedings of the XIX TARK Conference, DOI:\href{http://dx.doi.org/10.4204/EPTCS.379.2}{10.4204/EPTCS.379.2}. Financial support from the Yrj\"{o} Jahnsson Foundation (Grant number 20217398) and the OP Group Research Foundation (Grant numbers 20230084 and 20240089) is gratefully acknowledged. All errors are my own.}}
\author{Michele Crescenzi\\ \href{mailto:michele.crescenzi@helsinki.fi}{michele.crescenzi@helsinki.fi}\\
ORCID iD: \href{https://orcid.org/0000-0002-3841-5514}{0000-0002-3841-5514}}
\affil{University of Helsinki and Helsinki Graduate School of Economics\\
Helsinki, Finland}
%\date{This version: March 2023}

\newtheorem{proposition}{Proposition}
\newtheorem{theorem}{Theorem}
\newtheorem*{theorem*}{Theorem}
\newtheorem*{definition*}{Definition}
\newtheorem{assumption}{Assumption}
\newtheorem{definition}{Definition}
\newtheorem{lemma}{Lemma}
\newtheorem{claim}{Claim}
\newtheorem{remark}{Remark}
\newtheorem*{remark*}{Remark}
\newtheorem*{corollary*}{Corollary}

\maketitle
%\onehalfspacing

\begin{abstract}
We study distributed knowledge, which is what a group of agents come to know by communicating freely with each other and sharing everything they know. Knowledge is not necessarily partitional since agents may be boundedly rational. Specifically, it is possible that when an agent knows an event, she fails to know that she knows it; it is also possible that when an agent does not know an event, she fails to know that she does not know it. We show that there are two qualitatively different cases of how distributed knowledge is attained. In the first, distributed knowledge is determined by any group member who can replicate all the inferences that anyone else in the group makes. This finding is in line with the extant literature. In the second case, no member can replicate all the inferences made within the group. As a result, distributed knowledge is determined by any two group members who can jointly replicate what anyone else infers. This case can be interpreted as a form of wisdom-of-crowd effect and shows that, contrary to what is generally believed, distributed knowledge cannot always be reduced to the reasoning abilities of a single group member.

\bigskip

JEL CLASSIFICATION: C70, D82, D83
\bigskip

\end{abstract}

\newpage

\onehalfspacing

\section{Introduction}
\subsection{Background and motivating example}
What can a group of privately informed people know? There are two limiting cases. One is common knowledge. An event is common knowledge when it is public: Everyone in the group knows that event, everyone knows that everyone knows it, and so on. The other limiting case is distributed knowledge. An event is distributed knowledge if group members can come to know it by combining what each of them knows. For example, if a group member knows only an event $p$ and another member knows only that $p$ implies $q$, then knowledge of $q$ is distributed within the group. While \cite{mccarthy1978model} characterize common knowledge as what ``any fool'' knows, distributed knowledge ``can be viewed as what a `wise man'---one who has complete knowledge of what each member of the group knows---would know'' \citep[p. 3]{fagin2004reasoning}. Taken together, these two notions circumscribe all feasible configurations of group knowledge: A group cannot know less than what is common knowledge, and cannot know more than what is distributed knowledge among its members. 

The purpose of this paper is to model distributed knowledge for groups of agents who may have non-partitional knowledge because of their bounded rationality. As we are going to argue below, while common knowledge has been extended to non-partitional knowledge with relative ease, the extension of distributed knowledge to non-partitional models is significantly subtler.

In economics and game theory, it is customary to model group knowledge assuming that information partitions---or, equivalently, equivalence relations---represent individual knowledge. Common and distributed knowledge, then, are represented by the finest common coarsening (meet) and the coarsest common refinement (join) of the individual information partitions, respectively \citep{aumann1976agreeing, tobias2021meet}. When equivalence relations are employed, the transitive closure of the union of the individual relations represents common knowledge, whereas their intersection represents distributed knowledge \citep{fagin2004reasoning}. But to form partitions or equivalence relations presupposes logical omniscience and flawless information processing skills, thus requiring significant, if not exorbitant, cognitive resources \citep{Lipman-survey, rubinstein1998Book}. A less idealized description of individual rationality is offered by non-partitional models, which were introduced in the economics literature by  \cite{Geanakoplos1989NonPart, GeanakoplosNonPartition}. In non-partitional models of knowledge, people may fail to fully process their information. As a result, individual knowledge is not represented by a partition or an equivalence relation, but by a reflexive binary relation that is not necessarily symmetric or transitive. How is group knowledge represented in non-partitional models? As shown by \cite{Geanakoplos1989NonPart}, common knowledge is still represented by the transitive closure of the union of the individual binary relations. By contrast, the following example demonstrates that the usual representation of distributed knowledge as the intersection of individual relations is not fit for purpose in non-partitional models. Aiming at correcting this shortcoming, the formulation of distributed knowledge provided in this paper is suitable for non-partitional and partitional models alike.\footnote{We study distributed knowledge under the assumption that individual knowledge is still represented by reflexive binary relations. A more comprehensive analysis of distributed knowledge that dispenses with that assumption is provided by \cite{fukuda2019distributed}.}

This example is adapted from \citet[p. 363]{GeanakoplosNonPartition}. Consider a situation of incomplete information with two possible states of the world, $\omega_1$ and $\omega_2$. There are two persons, Carol ($c$) and Bob ($b$). Carol is fully rational while Bob is not. Their knowledge is represented by the knowledge operators $K_c$ and $K_b$ and, equivalently, by the possibility relations $P_c$ and $P_b$ in Table \ref{tab:MotEx}.

\begin{table}[htb]
	%	\caption{Global caption}
	\begin{minipage}[t]{.55\linewidth}
		\centering
		\vspace{0pt}
		\begin{tabular}{c | c | c}%\label{tab:MotExKnow}
			$A$ (Events) & $K_c(A)$& $K_b(A)$ \\
			\hline
			$\emptyset$ & $\emptyset$& $\emptyset$\\
			$\{\omega_1\}$ & $\emptyset$& $\{\omega_1\}$\\
			$\{\omega_2\}$ & $\emptyset$& $\emptyset$\\
			$\{\omega_1, \omega_2\}$ & $\{\omega_1, \omega_2\}$& $\{\omega_1, \omega_2\}$
		\end{tabular}
		%	\caption{Knowledge operators}
	\end{minipage}%
	\begin{minipage}[t]{.4\linewidth}
		\centering
		\vspace{0pt}
		\begin{tabular}{c | c | c}%\label{tab:MotExRela}
			$\omega$ (States) & $P_c(\omega)$& $P_b(\omega)$\\
			\hline
			$\omega_1$ & $\{\omega_1, \omega_2\}$& $\{\omega_1\}$\\
			$\omega_2$ & $\{\omega_1, \omega_2\}$& $\{\omega_1, \omega_2\}$
		\end{tabular}
		%	\caption{Possibility relations}
	\end{minipage}
	\caption{Knowledge operators (left) and possibility relations (right)\label{tab:MotEx}}
\end{table}

Knowledge is about events, which are subsets of the set of possible states $\{\omega_1, \omega_2\}$. For each state $\omega$, a person's possibility relation $P$ is a binary relation telling us the states $P(\omega)$ that the person considers possible when $\omega$ occurs. For a given event $A$, a person's knowledge operator $K$ tells us the states $K(A)$ in which the person knows that $A$ holds. A person knows $A$ in $\omega$ if $P(\omega) \subseteq A$. In the example, Carol knows only the event $\{\omega_1, \omega_2\}$, and she knows it in each state. Bob is better informed, since in $\omega_1$ he also knows the event $\{\omega_1\}$. Bob's possibility relation $P_b$ is not an equivalence relation, reflecting his lack of full rationality. Bob is not fully rational---specifically, he is not negatively introspective---because he fails to make the following inference: ``There are two states. In one of them I know that $\{\omega_1\}$ holds, in the other I do not. Therefore, when I do not know whether $\{\omega_1\}$ holds, the true state must be $\omega_2$. This means that I know the event $\{\omega_2\}$ in state $\omega_2$.'' Now, what is distributed knowledge between the fully rational Carol and the boundedly rational Bob? According to the standard representation, distributed knowledge is characterized by the intersection between $P_c$ and $P_b$, thus implying that Carol and Bob cannot come to know $\{\omega_2\}$. But if we take distributed knowledge to be what people come to know by sharing everything they know, knowledge of $\{\omega_2\}$ is arguably distributed in state $\omega_2$. If Bob shares all his knowledge with the fully introspective Carol, she will make the inference that not knowing $\{\omega_1\}$ is equivalent to knowing $\{\omega_2\}$, which is the inference that Bob failed to make on his own. Thus Carol can come to know the prevailing state even in $\omega_2$, and then she can talk back to Bob to share her newly acquired knowledge. In sum, perfect knowledge of the prevailing state is always attainable, and hence distributed, between Carol and Bob. Formally, the appropriate representation of distributed knowledge between our two agents is the identity knowledge operator on the set of all events, which in turn is represented by the possibility relation $P_D$ such that $P_D(\omega) = \{\omega\}$ for each state $\omega$.

\subsection{Research questions, contributions, and results}
The motivating example suggests that one cannot always rely on the standard representation of distributed knowledge as the intersection of possibility relations. This paper provides a new formulation of distributed knowledge that can accommodate non-partitional knowledge as well as heterogeneous reasoning skills. The two research questions we set out to answer are: What is the knowledge operator that describes distributed knowledge in a group where individual knowledge is not necessarily partitional? And what is the possibility relation representing that knowledge operator?

The non-partitional approach to knowledge enables us to incorporate a form of bounded rationality into the standard model of interactive knowledge. Specifically, we can dispense with the axioms of positive and negative introspection just by replacing partitions and equivalence relations with reflexive binary relations. The two introspection axioms are hardwired into partitional models and they both pertain to metacognition. Positive introspection says that, when an agent knows an event, she knows that she knows it; negative introspection says that, when she does not know an event, she knows that she does not know it. These axioms are far from being uncontroversial. Economists have criticized negative introspection as too strong \citep{Geanakoplos1989NonPart, shin1993logical}, whereas philosophers, most notably \cite{williamson2002knowledge}, have called both axioms into question (see \cite{Aucher2014} for an overview).

The ability to dispense with the two introspection axioms leads naturally to the study of groups in which, just like in the motivating example, different members may have different reasoning skills. We consider three kinds of agents: fully introspective, positively introspective, and non-introspective. They differ in the ability to form higher-order knowledge. A fully introspective agent satisfies both positive and negative introspection, and her knowledge is represented by an equivalence relation or, equivalently, a partition. In partitional models, everyone is fully introspective. By contrast, a positively introspective agent fails to satisfy negative introspection, and her knowledge is represented by a reflexive and transitive binary relation. Finally, a non-introspective agent fails to satisfy both positive and negative introspection, and her knowledge is represented by a reflexive binary relation.\footnote{As \cite{grabiszewski2015epistemic_EB} shows, an agent who does not satisfy an introspection axiom does not know \textit{whether} she satisfies that axiom or not.}

We define distributed knowledge as what group members come to know by sharing all their private knowledge. Knowledge sharing takes the form of a rational dialogue \textit{\`{a} la} \cite{geanakoplos1982} in which agents share everything they know in a non-strategic and honest manner. Because of possibly different levels of introspection, we have to model how people make inferences during communication. For instance, how would a fully introspective agent revise her knowledge upon learning what a positively introspective agent knows? In the motivating example we argued informally that the fully introspective Carol would always find out the prevailing state of the world. But a precise and formal description of inference making is needed if we want to investigate more complicated cases. Our paper shows that distributed knowledge is affected directly by how group members make inferences and revise their knowledge.

Our first contribution is to develop a model of inference making and knowledge revision. Specifically, we introduce revision operators and revision types to model how differently introspective agents think through the information they learn from others. A revision operator maps any knowledge operator to a revised knowledge operator and captures a specific type of inference; a revision type is formed out of revision operators and tells us the whole set of inferences an agent can draw. We show that revision types are partially ordered according to the precision of the knowledge they induce. Agents with higher revision types infer more detailed knowledge from the same input than do agents with lower types. Differently put, an agent with a higher revision type can replicate all the inferences that an agent with a lower type makes. The set of feasible revision types is a lattice. We show that there exist seven feasible revision types, and not just three---namely, one for each introspection level---as one would expect. Besides, certain types are not comparable: when two types are incomparable, neither one can replicate all the inferences that the other type can conceivably draw. This follows from the order dependence of inference making. As it turns out, revision operators do not necessarily commute with each other and, consequently, what an agent ultimately knows depends not only on her degree of introspection but also on the order in which she makes inferences. It may well be the case that two equally introspective agents process the very same body of knowledge but reach different conclusions just because of the different order in which they reason. When this happens, the two agents have different revision types. Both the order dependence of inference making and the incomparability of certain revision types are effects of non-partitional knowledge; they both vanish when knowledge is partitional.

We develop our model of inference making and knowledge revision to explain how distributed knowledge is attained. As \citet[p. 43]{rubinstein1998Book} aptly remarks, the standard model of interactive knowledge in economics and game theory ``has to be thought of as a reduced form derived from a more complete model, one that captures the decision maker's inference process.'' This paper contains our attempt to develop a more complete model of interactive knowledge that captures how people make inferences upon learning other people's knowledge. On the other hand, our model does not explain how people form the knowledge they have prior to communicating with others.

The core contribution of our paper is to offer a new formalization of distribute knowledge. Building on our model of inference making and knowledge revision, we establish that distributed knowledge is represented by the knowledge operator obtained by revising the initial knowledge in the group according to the join (least upper bound) of the agents' revision types. Furthermore, the communication process by which group members attain distributed knowledge is affected by how many maximal revision types exist in the group. Since the set of all feasible revision types is a finite lattice, a maximal type always exists. We show that only two cases are possible. In the first, there is only one maximal revision type in the group. This means that there is a member whose revision type is the join of all the types in the group, and such a member can replicate all the inferences that anyone else in the group makes. As a consequence, distributed knowledge can be attained by letting every group member communicate all her knowledge just to the member having the maximal type. This result is perfectly in line with the common view that distributed knowledge is ``just the knowledge of one distinct agent, to whom all the agents communicate their knowledge'' \citep[p. 217]{derHoek-et-alDistributed}. The agent with the highest type can be interpreted as the ``wise man'' we mentioned earlier. In the second case, there are exactly two maximal revision types in the group, so that no member has a revision type equal to the join of all types in the group. This also means that no one in the group can replicate all the inferences that anyone else makes. As a consequence, distributed knowledge can be attained in a two-stage procedure. First, everyone communicates her initial knowledge to the two agents with maximal revision types, and also the two maximal types share their initial knowledge with each other. Then, in the second stage it is sufficient that the two maximal types share with each other the knowledge they revised at the end of the first stage. We show that distributed knowledge coincides with what the more introspective of the two maximal types knows at the end of the second stage. Crucially, the knowledge shared by the two maximal types is different across stages. In the first stage, they share their own initial knowledge; in the second, they share what they came to know after having learned everyone else's knowledge and having thought it through. This result is at odds with the interpretation of distributed knowledge as the knowledge of a distinct agent to whom all other group members communicate their knowledge. In fact, distributed knowledge turns out to be the combined knowledge of two distinct agents, to whom everyone communicates her initial knowledge. The two group members with maximal revision types act jointly as the ``wise man.''

One can interpret distributed knowledge as wisdom of the crowd. While it is obvious that distributed knowledge is always more detailed than individual knowledge, our model highlights two peculiarities of distributed knowledge which emerge in groups having two maximal revision types. First, a group with just two maximal types attains the same distributed knowledge that would have been attained if their join had been in the group. This means that a group with lower revision types can be as wise as a group with a strictly higher type would be. However, the latter group would attain distributed knowledge more quickly. Second, distributed knowledge can be strictly more detailed than what any group member would come to know just by knowing everyone else's initial knowledge. We give an example of this in Subsection \ref{subsec:Example2}. These peculiarities stem from the incomparability of the two maximal revision types. The inferences of both types are incorporated in the distributed knowledge that the group attains through honest and unfettered communication. As a result, distributed knowledge cannot be reduced to the reasoning abilities of a single group member.

We use primarily knowledge operators to formalize inference making and knowledge revision. Since knowledge operators and possibility relations are isomorphic, we also formulate our main results in terms of relations. The possibility relation representing distributed knowledge is always contained in the intersection of the group members' individual possibility relations. The set inclusion can be strict, just like in the motivating example; it holds with equality only in certain cases, among which is the standard case in economics and game theory of fully introspective groups with partitional knowledge.

\subsection{Paper structure}
In Section \ref{Sec:Model} we introduce the basic ingredients of the model, namely knowledge operators and possibility relations. To formalize inference making, we introduce revision operators in Section \ref{sec:RevisionOperators} and study their properties. We put revision operators into use in Section \ref{Sec:KRevision}, where we define revision types and formalize how agents with different levels of introspection revise their knowledge during communication. In the same section, we also introduce a partial order on the set of revision types. We give a full characterization of distributed knowledge in Section \ref{Sec:DK}. Finally, in Section \ref{Sec:Disc} we discuss the related literature and some of the assumptions made in our model.

\section{Knowledge operators and possibility relations}\label{Sec:Model}
The basic components of our model are: a non-empty and finite set of states of the world $\Omega$; a finite set of agents $I = \{1, 2, \dots, N\}$, with $N \geq 2$; knowledge operators (Subsection \ref{subsec:KOperators}); and possibility relations (Subsection \ref{subsec:PRelations}). Knowledge operators and possibility relations represent agents' knowledge about events in $\Omega$. We confine our analysis to the case in which the two representations are equivalent (Subsection \ref{subsec:Duality}). Most of the material in this section is standard and taken from the literature on interactive epistemology---see, among many others, \cite{aumann1999interactive1}, \cite{morris1996logic}, \cite{rubinstein1998Book}, and \cite{Samet-Ignoring}. The only novelty here is the left $n$-ary trace introduced in Subsection \ref{subsec:PRelations}.

\subsection{Knowledge operators}\label{subsec:KOperators}
A knowledge operator $K$ is a function $K: 2^{\Omega} \to 2^{\Omega}$. Subsets of $\Omega$ are called events. An event $A\subseteq \Omega$ holds at a state $\omega\in \Omega$ if $\omega \in A$. If $K_i$ is agent $i$'s knowledge operator, then $K_i(A)$ is the event ``$i$ knows that $A$''. To wit, $K_i(A)$ is the set of states in which $i$ knows that $A$ holds.

A knowledge operator $K$ satisfies:
\begin{itemize}
\item \textit{distributivity} if for all $A,B \subseteq \Omega$,
\begin{equation}\label{eq:K0}
K(A\cap B) = K(A) \cap K(B); \tag{K0}
\end{equation}
\item \textit{necessitation} if
\begin{equation}\label{eq:K0'}
K(\Omega) = \Omega; \tag{K0$'$}
\end{equation}	
\item \textit{veridicality} if for all $A\subseteq \Omega$,
\begin{equation}\label{eq:K1}
K(A)\subseteq A; \tag{K1}
\end{equation}		
\item \textit{positive introspection} if for all $A\subseteq \Omega$,
\begin{equation}\label{eq:K2}
K(A)\subseteq K(K(A)); \tag{K2}
\end{equation}
\item \textit{negative introspection} if for all $A\subseteq \Omega$,
\begin{equation}\label{eq:K3}
\lnot K(A) \subseteq K(\lnot K(A)), \tag{K3}
\end{equation}
where $\lnot$ denotes the set complement operation.
\end{itemize}

The axioms above are the standard \textit{S5} axioms. Recall that negative introspection and veridicality jointly imply necessitation and positive introspection \citep[pp. 269-270]{aumann1999interactive1}. Besides, distributivity implies monotonicity, which is defined as follows. An operator $K$ is \textit{monotone} if for all $A, B \subseteq \Omega$,
\begin{equation}\label{eq:Monot}
A \subseteq B \implies K(A) \subseteq K(B).
\end{equation}

Call $\mathcal{K}$ the set of all knowledge operators on $\Omega$. We write $\mathcal{K}(\Omega)$ when we need to specify the underlying state space. The following subsets of $\mathcal{K}$ are introduced for later use:
\begin{itemize}
\item $\mathcal{K}^v$ is the set of all knowledge operators that satisfy \eqref{eq:K1};
	\item $\mathcal{K}^0$ is the set of all knowledge operators that satisfy \eqref{eq:K0} and \eqref{eq:K0'};
	\item $\mathcal{K}^1$ is the set of all knowledge operators that satisfy \eqref{eq:K0}, \eqref{eq:K0'}, and \eqref{eq:K1};
	\item $\mathcal{K}^2$ is the set of all knowledge operators that satisfy \eqref{eq:K0}, \eqref{eq:K0'}, \eqref{eq:K1}, and \eqref{eq:K2};
	\item $\mathcal{K}^3$ is the set of all knowledge operators that satisfy \eqref{eq:K0}, \eqref{eq:K0'}, \eqref{eq:K1}, \eqref{eq:K2}, and \eqref{eq:K3}.
\end{itemize}
It is clear that $\mathcal{K}^3 \subseteq \mathcal{K}^2 \subseteq \mathcal{K}^1 \subseteq \mathcal{K}^0 \subseteq \mathcal{K}$. 

\subsection{Possibility relations}\label{subsec:PRelations}
A possibility relation $P$ is a binary relation over $\Omega$, i.e., $P \subseteq \Omega \times \Omega$. If $(\omega, \omega')$ is in agent $i$'s possibility relation, then $i$ considers $\omega'$ possible when $\omega$ occurs. It is often convenient to describe a possibility relation with its lower contour sets. Formally, the lower contour set of $P$ at $\omega$ is the set $P(\omega):= \left\lbrace \omega' \in \Omega: (\omega, \omega')\in P\right\rbrace$. In words, $P(\omega)$ is the set of states that, according to $P$, are considered possible when $\omega$ occurs.

A possibility relation $P$ is:
\begin{itemize}
\item \textit{reflexive} if for all $\omega \in \Omega$,
\begin{equation*}
(\omega, \omega) \in P;
\end{equation*}
\item \textit{transitive} if for all $\omega, \omega', \omega'' \in \Omega$,
\begin{equation*}
(\omega, \omega') \in P \text{ and } (\omega', \omega'') \in P \implies (\omega, \omega'') \in P;
\end{equation*}
\item \textit{Euclidean} if for all $\omega, \omega', \omega'' \in \Omega$,
\begin{equation*}
(\omega, \omega') \in P \text{ and } (\omega, \omega'') \in P \implies (\omega', \omega'') \in P;
\end{equation*}
\item \textit{symmetric} if for all $\omega, \omega' \in \Omega$,
\begin{equation*}
(\omega, \omega') \in P \implies (\omega', \omega) \in P.
\end{equation*}
\end{itemize}

Call $\mathcal{P}$ the set of all possibility relations on $\Omega$. The following subsets of $\mathcal{P}$ will be used later on:
\begin{itemize}
\item $\mathcal{P}^1$ is the set of all \textit{reflexive} relations on $\Omega$;
\item $\mathcal{P}^2$ is the set of all \textit{reflexive} and \textit{transitive} relations on $\Omega$;
\item $\mathcal{P}^3$ is the set of all \textit{reflexive}, \textit{transitive} and \textit{Euclidean} relations on $\Omega$.
\end{itemize}

It is clear that $\mathcal{P}^3 \subseteq \mathcal{P}^2 \subseteq \mathcal{P}^1 \subseteq \mathcal{P}$. It is also easy to check that $P \in \mathcal{P}^3$ if and only if $P$ is an equivalence relation, i.e., a reflexive, symmetric and transitive relation. A similar taxonomy of possibility relations with a discussion of their properties can be found in \cite{GeanakoplosNonPartition}.

Some of the results in Section \ref{Sec:KRevision} involve two particular binary relations: the left trace and the left $n$-ary trace. The left trace is a relation derived from another binary relation. Formally, following \citet[p. 69]{aleskerov2007utility}, the \textit{left trace} of $P \in \mathcal{P}$ is the binary relation $T_P$ constructed as follows:
\begin{equation*}
\text{For all } \omega, \omega' \in \Omega, \quad (\omega, \omega') \in T_P \iff P(\omega') \subseteq P(\omega).
\end{equation*}

That is, $(\omega, \omega')$ is in $T_P$ if and only if every state considered possible when $\omega'$ occurs is also considered possible when $\omega$ occurs. It is easy to check that the left trace is reflexive and transitive. The symmetric part of $T_P$ is the binary relation $E_P$ obtained as: 
\begin{equation*}
\text{For all } \omega, \omega' \in \Omega, \quad (\omega, \omega') \in E_P \iff (\omega, \omega') \in T_P \text{ and } (\omega', \omega) \in T_P.
\end{equation*}

Notice that $(\omega, \omega') \in E_P$ if and only if $P(\omega') = P(\omega)$. 

The left $n$-ary trace is a generalization of the left trace, and it is formed out of $n$ other binary relations. To the best of our knowledge, this is the first paper to introduce the left $n$-ary trace. Formally, for a given $n$-tuple $\mathbf{P} = (P_1, \dots, P_n)$ of relations in $\mathcal{P}$, 
the \textit{left $n$-ary trace} of $\mathbf{P}$ is the binary relation $T_{\mathbf{P}}$ constructed as follows:
\begin{align*}
\text{For all } \omega, \omega' \in \Omega, \quad (\omega, \omega') \in T_{\mathbf{P}} \iff \; &\text{ for all } i\in \{1,\dots, n\} \; \text{ there exists a } \; j\in \{1,\dots, n\} \\
&\text{ such that } \;  P_j(\omega') \subseteq P_i(\omega).
\end{align*}

Differently put, $(\omega, \omega')$ is in $T_{\mathbf{P}}$ if and only if, for each relation $P_i$ in $\mathbf{P}$, we can find a relation $P_j$ in $\mathbf{P}$, not necessarily different from $P_i$, that satisfies the following property: every state considered possible according to $P_j$ when $\omega'$ occurs is also considered possible according to $P_i$ when $\omega$ occurs. It is easy to check that the left $n$-ary trace is reflexive and transitive. Clearly, $n=1$ and $T_\mathbf{P} = T_P$ when $\mathbf{P} = (P)$. The symmetric part of $T_{\mathbf{P}}$ is denoted $E_{\mathbf{P}}$ and is defined in the obvious way:
\begin{equation*}
\text{For all } \omega, \omega' \in \Omega, \quad (\omega, \omega') \in E_{\mathbf{P}} \iff (\omega,\omega') \in T_{\mathbf{P}} \; \text{ and } \;  (\omega',\omega) \in T_{\mathbf{P}}.
\end{equation*}

In Appendix \ref{MinimalSets} we provide a characterization that will prove useful for computing $E_{\mathbf{P}}$. Table \ref{tab:Traces} contains an example with two possibility relations $P_1$ and $P_2$ on $\Omega=\{\omega_1, \dots, \omega_4\}$, the corresponding left traces $T_{P_1}$ and $T_{P_2}$, the left binary trace $T_{\mathbf{P}}$, where $\mathbf{P} = (P_1, P_2)$, and its symmetric part $E_{\mathbf{P}}$. To illustrate how the left binary trace is obtained, notice that $(\omega_1, \omega_2) \in T_{\mathbf{P}}$ because both $P_2(\omega_2) \subseteq P_1(\omega_1)$ and $P_2(\omega_2) \subseteq P_2(\omega_1)$ hold. By contrast, to conclude that $(\omega_1, \omega_3) \notin T_{\mathbf{P}}$, it is enough to observe that neither $P_1(\omega_3) \subseteq P_1(\omega_1)$ nor $P_2(\omega_3) \subseteq P_1(\omega_1)$ is true.

\begin{table}[h]
\centering
\begin{tabular}{ c | c | c | c | c | c | c}
$\omega$ &  $P_1(\omega)$  &  $P_2(\omega)$ &  $T_{P_1}(\omega)$ &  $T_{P_2}(\omega)$ &  $T_{\mathbf{P}}(\omega)$ &  $E_{\mathbf{P}}(\omega)$\\
  \hline			
  $\omega_1$ & $\{\omega_1,\omega_2\}$ & $\{\omega_1, \omega_2\}$ & $\{\omega_1\}$ & $\{\omega_1, \omega_2\}$ & $\{\omega_1, \omega_2\}$ & $\{\omega_1\}$\\
  $\omega_2$ & $\{\omega_2, \omega_3\}$ & $\{\omega_2\}$ & $\{\omega_2\}$ & $\{\omega_2\}$ & $\{\omega_2\}$ & $\{\omega_2\}$\\
  $\omega_3$ & $\{\omega_2, \omega_3, \omega_4\}$ & $\{\omega_2, \omega_3, \omega_4\}$ & $\{\omega_2, \omega_3, \omega_4\}$& $\{\omega_2, \omega_3\}$ & $\{\omega_2, \omega_3, \omega_4\}$ & $\{\omega_3\}$\\
  $\omega_4$ & $\{\omega_4\}$ & $\{\omega_1, \omega_4\}$ & $\{\omega_4\}$& $\{\omega_4\}$& $\{\omega_4\}$ & $\{\omega_4\}$
\end{tabular}
\caption{The left traces and the left binary trace of $P_1$ and $P_2$.}
\label{tab:Traces}
\end{table}

\subsection{The duality between knowledge operators and possibility relations}\label{subsec:Duality}

Every possibility relation represents a knowledge operator, but only knowledge operators satisfying both distributivity \eqref{eq:K0} and necessitation \eqref{eq:K0'} can be represented by a possibility relation \citep{morris1996logic}. Throughout this paper we confine ourselves to knowledge operators that are representable by a possibility relation.

Given a possibility relation $P \in \mathcal{P}$, the unique knowledge operator represented by $P$ is 
\begin{equation}\label{eq:ReprKbyP}
K(A) = \left\lbrace \omega \in \Omega:  P(\omega) \subseteq A\right\rbrace
\end{equation}
for all events $A \subseteq \Omega$. Conversely, given a knowledge operator $K \in \mathcal{K}^0$, the unique possibility relation that represents $K$ is  
\begin{equation}\label{eq:ReprMSZ}
	P(\omega) := \bigcap \left\lbrace A \subseteq \Omega : \omega \in K(A)\right\rbrace
\end{equation}
for all $\omega \in \Omega$. In words, \eqref{eq:ReprMSZ} says that the states considered possible when $\omega$ occurs are exactly those which are common to all events known in $\omega$.

In light of \eqref{eq:ReprKbyP} and \eqref{eq:ReprMSZ}, knowledge operators and possibility relations are interchangeable provided that the former satisfy distributivity and necessitation. Formally, there is a dual isomorphism\footnote{The terminology is borrowed from \cite{monjardet2007dual}.} between $\mathcal{P}$ and $\mathcal{K}^0$. To establish this, notice that $\mathcal{P}$ is partially ordered by set inclusion. The set $\mathcal{K}$ is partially ordered by pointwise set inclusion and so is its subset $\mathcal{K}^0$. Specifically, for any $K, K'\in \mathcal{K}$,
\begin{equation*}
K \leq K' \iff K(A) \subseteq K'(A) \text{ for all } A\subseteq \Omega.
\end{equation*}

In words, when $K \leq K'$ a person with knowledge $K'$ knows at least as much as someone with knowledge $K$ does. We also say that $K'$ is an expansion of $K$ or that $K'$ is more detailed than $K$.

\begin{lemma}\label{lemma:Duality}
Let $f: \mathcal{P} \to \mathcal{K}^0$ be the function that associates to each relation $P$ the knowledge operator $f(P)$ formed according to \eqref{eq:ReprKbyP}. Then the function $f$ is a dual isomorphism. That is, $f$ is a bijection such that
\begin{equation*}
P \subseteq P' \iff f(P') \leq f(P)
\end{equation*}
for all $P, P' \in \mathcal{P}$.
\end{lemma}
\begin{proof}
See Appendix \ref{Proof:Duality}.
\end{proof}

\cite{Samet-Ignoring}, \cite{morris1996logic}, and \citet[Exercise 71.1]{ORgametheory94} contain results that are partly analogous to the lemma above. The dual isomorphism between knowledge operators and possibility relations entails that each property of the former corresponds to a specific property of the latter, and vice versa.

\begin{lemma}[\cite{morris1996logic}] Let $K$ be a knowledge operator in $\mathcal{K}^0$ and let $P \in \mathcal{P}$ be the possibility relation that represents $K$. Then the following hold.
\begin{itemize}
\item[(i)] $K \in \mathcal{K}^1$ if and only if $P \in \mathcal{P}^1$.
\item[(ii)] $K \in \mathcal{K}^2$ if and only if $P \in \mathcal{P}^2$.
\item[(iii)] $K \in \mathcal{K}^3$ if and only if $P \in \mathcal{P}^3$.
\end{itemize}
\end{lemma}

\section{Revision operators}\label{sec:RevisionOperators}
The notion of distributed knowledge we adopt presupposes a process of communication and inference making. This section introduces revision operators in order to capture inference making. A revision operator is a map $(.)^{\mathbf{r}}:\mathcal{K}^v \to \mathcal{K}^v$ that assigns a new knowledge operator $K^{\mathbf{r}}$ to any given knowledge operator $K$. The operator $K^{\mathbf{r}}$ is a revision of $K$ based on a specific type of inference. We define three revision operators. The first describes inferences based on positive introspection \eqref{eq:K2}. The second operator is based on both positive \eqref{eq:K2} and negative introspection \eqref{eq:K3}, and the last one captures distributivity \eqref{eq:K0}.

A revision operator describes just a specific type of inference that can be drawn in knowledge revision. An agent may draw different inferences described by different revision operators. The exact set of inferences made by the agent is represented by her revision type. Revision types are introduced in Section \ref{Sec:KRevision}.

A couple of definitions before getting to revision operators. Given a knowledge operator $K \in \mathcal{K}$, the image of $K$ is
\begin{equation*}
\mathsf{Img}(K):= \left\lbrace K(A): A \subseteq \Omega\right\rbrace,
\end{equation*}
and the set of complements of the elements of $\mathsf{Img}(K)$ is
\begin{equation*}
\mathsf{Img}(\lnot K):= \left\lbrace \lnot K(A): A \subseteq \Omega\right\rbrace.
\end{equation*}

\subsection{Positive Introspection}
The positive introspection operator represents inferences based on the positive introspection axiom \eqref{eq:K2}. Formally, the \textit{positive introspection operator} is a function $(.)^{\bm{+}}: \mathcal{K}^v \to \mathcal{K}^v$ that maps any $K$ to the revised knowledge operator $K^{\bm{+}}$ constructed as follows. For all $A \subseteq \Omega$,
\begin{equation*}
K^{\bm{+}} (A) = 
\begin{cases}
A & \text{ if } A \in \mathsf{Img}(K)\\
K(A) & \text{ if } A \notin \mathsf{Img}(K).
\end{cases}
\end{equation*}

In words, if $A = K(B)$ for some event $B$, then knowledge of $B$ is equivalent to $A$. Thus, a positively introspective agent realizes that her knowing $B$ is equivalent to knowing $A$. As a result, her revised knowledge about $A$ is $K^{\bm{+}}(A) = A$. If $A \neq K(B)$ for all $B \subseteq \Omega$, then there is no event $K(B)$ equivalent to $A$. Hence the agent's knowledge of $A$ is left unchanged and $K^{\bm{+}}(A) = K(A)$.

As an example of the reasoning behind the positive introspection operator, consider an agent whose knowledge at a certain point in time is described by $K$ in Table \ref{tab:PosInt}. The state space is $\Omega=\{\omega_1, \omega_2, \omega_3\}$. If the agent were positively introspective, she would pause a moment to reflect upon her knowledge and reason along the following lines. ``I never know when $\{\omega_2\}$ holds. But I know that $\{\omega_1, \omega_2\}$ holds at $\omega_2$, and only at $\omega_2$. Therefore, when I am at a state in which I know that $\{\omega_1, \omega_2\}$ holds, I can conclude that the state must be $\omega_2$. Thus I know when $\{\omega_2\}$ holds.'' A positively introspective agent could follow such a line of reasoning either before or after the true state of the world has occurred. In the former case, the agent ponders what she would know at every possible state of the world, so thinking about her future knowledge. In the other case, the agent reflects both on her factual and counterfactual knowledge. Assuming $\omega_2$ has occurred, the realization that $K^{\bm{+}}(\{\omega_2\}) = \{\omega_2\}$ is a piece of factual knowledge. The realization that $K^{\bm{+}}(\{\omega_3\}) = \{\omega_3\}$ refers to what the agent would have known if $\omega_3$ had occurred.

\begin{table}[h]
\centering
\begin{tabular}{ c | c | c | c}
$A$ &  $K(A)$ & $K^{\bm{+}} (A)$ & $K^{\bm{\pm}}(A)$\\
\hline			
$\emptyset$ & $\emptyset$ & $\emptyset$ & $\emptyset$  \\
$\{\omega_1\}$ & $\emptyset$ & $\emptyset$ & $\emptyset$\\
$\{\omega_2\}$ & $\emptyset$ & $\{\omega_2\}$ & $\{\omega_2\}$\\
$\{\omega_3\}$ & $\emptyset$ & $\{\omega_3\}$ & $\{\omega_3\}$\\
$\{\omega_1,\omega_2\}$ & $\{\omega_2\}$ & $\{\omega_2\}$ & $\{\omega_1,\omega_2\}$\\
$\{\omega_1,\omega_3\}$ & $\emptyset$ & $\emptyset$ & $\{\omega_1,\omega_3\}$\\
$\{\omega_2,\omega_3\}$ & $\{\omega_3\}$ & $\{\omega_3\}$ & $\{\omega_3\}$\\
$\{\omega_1,\omega_2,\omega_3\}$ & $\{\omega_1,\omega_2,\omega_3\}$ & $\{\omega_1,\omega_2,\omega_3\}$ & $\{\omega_1,\omega_2,\omega_3\}$
\end{tabular}
\caption{A knowledge operator $K$ and the corresponding revised operators $K^{\bm{+}}$ and $K^{\bm{\pm}}$.}
\label{tab:PosInt}
\end{table}

The main properties of the positive introspection operator are formally stated and proved in Appendix \ref{Proof:PosIntrProperties}. It is worth noting here that $(.)^{\bm{+}}$ is not a closure operator, since it is extensive and idempotent but not monotone. Extensivity means that $K^{\bm{+}}$ is an expansion of $K$, i.e., no piece of knowledge is lost or forgotten when changing $K$ into $K^{\bm{+}}$. Idempotence means that $K^{\bm{+}}$ cannot be further expanded through inferences rooted in the positive introspection axiom. Lack of monotonicity\footnote{The monotonicity mentioned here should not be confused with the monotonicity in \eqref{eq:Monot}. The former refers to an operator that maps knowledge operators to knowledge operators. The latter refers to a single knowledge operator that maps events to events.} means that, if a knowledge operator $J$ is an expansion of another operator $K$, then it is not necessarily true that $J^{\bm{+}}$ is an expansion of $K^{\bm{+}}$. A counterexample to the monotonicity of $(.)^{\bm{+}}$ can be found in Appendix \ref{LackMonot}.

In conclusion, we can interpret the positive introspection operator in terms of self-evident events. Recall that an event $A$ is self-evident \footnote{Following \cite{fukuda2019epistemic}, we define self-evident events relative to a single knowledge operator, whereas \cite{aumann1999interactive1} calls an event self-evident if it is self-evident (as per our definition) for all agents in $I$. In addition, notice that, for any knowledge operator satisfying veridicality, self-evident events and fixed points are the same.} for a knowledge operator $K$ if $A \subseteq K(A)$. One can check that $\mathsf{Img}(K)$ coincides with the set of self-evident events of $K^{\bm{+}}$, and that $K^{\bm{+}}$ is the least expansion of $K$ that turns all the elements of $\mathsf{Img}(K)$ into self-evident events.

\subsection{Full Introspection}
The full introspection operator describes inferences rooted in negative introspection \eqref{eq:K3} and, consequently, positive introspection \eqref{eq:K2}. Formally, the \textit{full introspection operator} is a function $(.)^{\bm{\pm}}: \mathcal{K}^v \to \mathcal{K}^v$ that maps any $K$ to the revised knowledge operator $K^{\bm{\pm}}$ constructed as follows. For all $A \subseteq \Omega$,
\begin{equation*}
K^{\bm{\pm}} (A) = 
\begin{cases}
A & \text{ if } A \in \mathsf{Img}(K)\cup \mathsf{Img}(\lnot K)\\
K(A) & \text{otherwise}.
\end{cases}
\end{equation*}

It is clear that $(.)^{\bm{\pm}}$ subsumes $(.)^{\bm{+}}$. The difference between the two is that the full introspection operator also takes the following into account. If $A = \lnot K(B)$ for some event $B$, then lack of knowledge of $B$ is equivalent to $A$. Thus, a fully introspective agent realizes that her not knowing  $B$ is equivalent to knowing $A$. Hence, her revised knowledge about $A$ is $K^{\bm{\pm}} (A) = A$.

As an example, consider again an agent with knowledge $K$ in Table \ref{tab:PosInt}. If the agent were fully introspective, she would think through her knowledge as follows. ``I know that $\{\omega_1, \omega_2\}$ holds at $\omega_2$, and only at $\omega_2$. This means that I do not know $\{\omega_1, \omega_2\}$ when either $\omega_1$ or $\omega_3$ occurs. Therefore, when I am at a state in which I do not know $\{\omega_1, \omega_2\}$, I can conclude that the state must be either $\omega_1$ or $\omega_3$. So I know when $\{\omega_1, \omega_3\}$ occurs.''

The main properties of the full introspection operator are formally stated and proved in Appendix \ref{Proof:FullIntrProperties}. Just like $(.)^{\bm{+}}$, the full introspection operator $(.)^{\bm{\pm}}$ is not a closure operator in that it is extensive and idempotent but fails to be monotone (see Appendix \ref{LackMonot}). In terms of self-evident events, it is easy to check that $\mathsf{Img}(K) \cup \mathsf{Img}(\lnot K)$ coincides with the set of self-evident events of $K^{\bm{\pm}}$, and that $K^{\bm{\pm}}$ is the least expansion of $K$ that turns all the elements of $\mathsf{Img}(K) \cup \mathsf{Img}(\lnot K)$ into self-evident events.

\subsection{Distributive closure}
The distributive closure represents inferences based on the distributivity axiom \eqref{eq:K0}. Formally, the \textit{distributive closure} is a function $(.)^{\mathbf{d}}: \mathcal{K}^v \to \mathcal{K}^v$ that maps any $K$ to the revised knowledge operator $K^{\mathbf{d}}$ constructed as follows. For all $A \subseteq \Omega$,
\begin{equation*}
K^{\mathbf{d}} (A) = \bigcup \left\lbrace \cap_{i=1}^n K(B_i): \; B_1, \dots, B_n \subseteq \Omega, \; \cap_{i=1}^n B_i \subseteq A, \; n \geq 1 \right\rbrace.
\end{equation*}

In words, if $A$ is jointly implied by a sequence of events $B_1, \dots, B_n$, then simultaneous knowledge of all these $n$ events implies knowledge of $A$.

As an example of the reasoning behind the distributive closure, consider a case with five states of the world $\omega_1, \dots, \omega_5$. Suppose an agent's knowledge is represented by $K$ in Table \ref{tab:DistrClos}. The agent always knows the state space but never knows any other event that is not reported in the first column of the table. If her knowledge satisfied distributivity, the agent would be able to draw the following inferences. ``I never know the event $\{\omega_2, \omega_5\}$. But this event holds if and only if the three events $\{\omega_1,\omega_2, \omega_3, \omega_5\}$, $\{\omega_2,\omega_3, \omega_4, \omega_5\}$ and $\{\omega_1,\omega_2, \omega_4, \omega_5\}$ simultaneously hold. Now, when $\omega_2$ occurs, and only in that case, I know that all the three events above hold. Therefore, I can conclude that also $\{\omega_2, \omega_5\}$ holds at $\omega_2$. And if I know that $\{\omega_2, \omega_5\}$ holds, I can infer that every superset of it holds as well.'' 

\begin{table}[h]
\centering
\begin{tabular}{ c | c | c}
$A$ &  $K(A)$ &  $K^{\mathbf{d}}(A)$\\
\hline			
$\{\omega_2, \omega_5\}$ & $\emptyset$ & $\{\omega_2\}$\\
$\{\omega_1,\omega_2, \omega_3, \omega_5\}$ & $\{\omega_1,\omega_2, \omega_3\}$ & $\{\omega_1,\omega_2, \omega_3\}$\\
$\{\omega_2,\omega_3, \omega_4, \omega_5\}$ & $\{\omega_2,\omega_3, \omega_4\}$ & $\{\omega_2,\omega_3, \omega_4\}$\\
$\{\omega_1,\omega_2, \omega_4, \omega_5\}$ & $\{\omega_1,\omega_2, \omega_4\}$ & $\{\omega_1,\omega_2, \omega_4\}$\\
$\dots$ & $\dots$ & $\dots$
\end{tabular}
\caption{A portion of a knowledge operator $K$ and of its distributive closure $K^{\mathbf{d}}$.}
\label{tab:DistrClos}
\end{table}

As we show in Appendix \ref{Proof:DistrClosProperties}, the distributive closure is a closure operator; as a result, $K^{\mathbf{d}}$ is the least expansion of $K$ satisfying distributivity. In the same appendix we also show that the distributive closure preserves both positive and negative introspection. By contrast, it is easy to check that neither of the two introspection operators preserves distributivity. This difference between the distributive closure and the two introspection operators plays an important role in defining revision types in the next section.

\section{Knowledge revision}\label{Sec:KRevision}
This section formalizes how agents revise their knowledge upon learning what others know. Consider an agent $i$ whose knowledge at a certain point in time is represented by $K_i$. Suppose $n-1$ other agents communicate their knowledge operators $K_1, \dots, K_{i - 1}, K_{i + 1}, \dots, K_n$ to $i$, as they are assumed to do in the next section. The question we address here is: How does $i$ think through the profile $\mathbf{K} = (K_1, \dots, K_n)$ and form a new knowledge operator out of it?

We make use of revision operators to represent knowledge revision. Different agents may use different revision operators depending on how introspective they are. An agent's revision type tells us which operators the agent uses as well as the order in which she employs them. Before formally introducing revision types in Subsection \ref{subsec:RevisionTypes}, we identify all feasible outcomes of knowledge revision in Subsection \ref{subsec:FeasibleRevisions}. In Subsection \ref{subsec:CharacPossibilityRelations}, we provide a full characterization of knowledge revision in terms of possibility relations. In Subsection \ref{subsec:PartialOrderTypes} we compare revision types by defining a partial order over them.

\subsection{Feasible revisions}\label{subsec:FeasibleRevisions}
The input to knowledge revision is a profile $\mathbf{K} = (K_1, \dots, K_n)$ of $n\geq 1$ knowledge operators; the output for agent $i\in \{1,\dots, n\}$ is a new knowledge operator denoted by $\mathbf{K}^{\theta_i}$ or $(K_1, \dots, K_n)^{\theta_i}$, where $\theta_i$ is agent $i$'s revision type. We define revision types in the next subsection; for the time being it is enough to know that any agent is one out of several revision types.

Since the input to knowledge revision is a profile of knowledge operators, agents revise their entire body of knowledge, i.e., they revise what they know or would know at any possible state of the world.\footnote{This form of comprehensive revision is analogous to the rational updating of information partitions introduced by \cite{weyers1992} in the literature of rational dialogues \citep{geanakoplos1982}. In a nutshell, Weyers argues that rational updating of an information partition should involve all the cells of the partition itself. Our model extends Weyers's idea to knowledge operators.} The assumption below spells out how knowledge revision takes place. To simplify notation, from now on we write $\cup_i K_i$ and $\cap_i K_i$ when the underlying index set is clear from the context.

\begin{assumption}\label{Assm:Rev}
Let $\mathbf{K} = (K_1, \dots, K_n)$ be a profile of $n\geq 1$ knowledge operators in $\mathcal{K}^1$. Agent $i\in \{1, \dots, n\}$ forms the knowledge operator $\mathbf{K}^{\theta_i}$ out of $\mathbf{K}$ as follows.
\begin{enumerate} 
\item[(i)] Agent $i$'s level of introspection determines which revision operators she employs to process $\mathbf{K}$. In particular:
\begin{enumerate}
\item[(a)] If $i$ is non-introspective, then she employs only the distributive closure $(.)^{\mathbf{d}}$ to form $\mathbf{K}^{\theta_i}$;
\item[(b)] If $i$ is positively introspective, then she employs only the positive introspection operator $(.)^{\bm{+}}$ and the distributive closure $(.)^{\mathbf{d}}$ to form $\mathbf{K}^{\theta_i}$;
\item[(c)] If $i$ is fully introspective, then she employs only the full introspection operator $(.)^{\bm{\pm}}$ and the distributive closure $(.)^{\mathbf{d}}$ to form $\mathbf{K}^{\theta_i}$.
\end{enumerate}
\item[(ii)] The knowledge operator $\mathbf{K}^{\theta_i}$ is such that $(\mathbf{K}^{\theta_i})^{\mathbf{r}} = \mathbf{K}^{\theta_i}$ for all revision operators $(.)^{\mathbf{r}}$ employed by $i$.
\item[(iii)] Agent $i$ forms $\mathbf{K}^{\theta_i}$ in one of two ways:
\begin{enumerate}
\item[(a)] Agent $i$ applies her revision operators to $\cup_j K_j$, which is the pointwise union\footnote{For all events $A$, the value of $\cup_{j} K_j$ at $A$ is $(\cup_{j} K_j) (A) =\cup_{j} K_j (A)$.} of the $n$ operators in $\mathbf{K}$;
\item[(b)] Agent $i$ acts sequentially. First, she applies her revision operators to each $K_j$ in $\mathbf{K}$. Revision operators are applied to each $K_j$ in the same order, and each $K_j$ is revised so that no further strict expansion is possible with the revision operators employed by $i$. Subsequently, $i$ takes the pointwise union of the $n$ operators just revised.
\end{enumerate}
\end{enumerate}
\end{assumption}

Part (i) says that an agent's level of introspection is reflected by the revision operators she adopts. Everyone employs the distributive closure. Part (ii) says that agents revise knowledge as much as possible, given the revision operators they adopt. For example, $\mathbf{K}^{\theta_i}$ can be the revised knowledge operator of a positively introspective $i$ only if it cannot be strictly expanded by applying the distributive closure or the positive introspection operator. As for part (iii), there are two ways in which an agent $i$ can form $\mathbf{K}^{\theta_i}$ out of the profile $\mathbf{K}$. The agent can aggregate all the $n$ operators in $\mathbf{K}$ and then apply the relevant revision operators. Alternatively, the agent can process each $K_j$ in $\mathbf{K}$ separately before aggregating them all. In this regard our analysis is positive but not normative: we do not take a stand on the ``correct'' order in which revision operators are to be used.

The next proposition lists all the feasible outcomes of knowledge revision. From now on, we abuse notation by omitting parenthesis when two or more revision operators are composed. For example, we write $(\cup_{j} K_j)^{\bm{+} \mathbf{d}}$ instead of the more accurate $((\cup_{j} K_j)^{\bm{+}})^{\mathbf{d}}$.

\begin{proposition}\label{Prop:FeasRev}
Let $\mathbf{K} = (K_1, \dots, K_n)$ be a profile of $n\geq 1$ knowledge operators in $\mathcal{K}^1$. Suppose an agent $i\in \{1, \dots, n\}$ revises $\mathbf{K}$ as per Assumption \ref{Assm:Rev}. Then the following hold.
\begin{enumerate}
\item[(i)] If agent $i$ is non-introspective, then $
\mathbf{K}^{\theta_i} = (\cup_{j} K_j)^{\mathbf{d}}$.
\item[(ii)] If agent $i$ is positively introspective, then $\mathbf{K}^{\theta_i}$ is equal to $(\cup_{j} K_j)^{\bm{+} \mathbf{d}}$ or $(\cup_{j} K_j)^{\mathbf{d}\bm{+} \mathbf{d}}$ if $i$ revises $\mathbf{K}$ according to part (iii-a) of Assumption \ref{Assm:Rev}, and it is equal to $(\cup_{j} K_j^{\bm{+}\mathbf{d}})^{\mathbf{d}}$ if knowledge is revised according to part (iii-b).
\item[(iii)] If agent $i$ is fully introspective, then $\mathbf{K}^{\theta_i}$ is equal to $(\cup_{j} K_j)^{\bm{\pm} \mathbf{d}}$ or $(\cup_{j} K_j)^{\mathbf{d}\bm{\pm} \mathbf{d}}$ if $i$ revises $\mathbf{K}$ according to part (iii-a) of Assumption \ref{Assm:Rev}, and it is equal to $(\cup_{j} K_j^{\bm{\pm}\mathbf{d}})^{\mathbf{d}}$ if knowledge is revised according to part (iii-b).
\end{enumerate}
\end{proposition}
\begin{proof}
See Appendix \ref{Proof:FeasRev}.
\end{proof}

The proposition above says that the outcome of knowledge revision is unique for a non-introspective agent: it makes no difference if the agent revises knowledge following part (iii-a) or part (iii-b) of Assumption \ref{Assm:Rev}. Things are different when an agent is positively introspective. In this case, the outcome of knowledge revision is $(\cup_{j} K_j)^{\bm{+} \mathbf{d}}$ or $(\cup_{j} K_j)^{\mathbf{d}\bm{+} \mathbf{d}}$ if the agent follows part (iii-a) of Assumption \ref{Assm:Rev}, whereas it is $(\cup_{j} K_j^{\bm{+}\mathbf{d}})^{\mathbf{d}}$ if part (iii-b) is followed. Importantly, the three knowledge operators just mentioned need not be equal to each other when at least one of the input knowledge operators in $\mathbf{K}$ is non-partitional. This means that knowledge revision is order dependent: two equally introspective agents may reach different conclusions when revising the very same profile of knowledge operators just because of the different order in which they apply revision operators and merge the input knowledge operators in $\mathbf{K}$. An instance of order dependence can be found in the example in Subsection \ref{subsec:Example1}, in which $(\cup_{j} K_j)^{\bm{+} \mathbf{d}}$ and $(\cup_{j} K_j)^{\mathbf{d}\bm{+} \mathbf{d}}$ differ (Table \ref{tab:ExKnow2}). The case of fully introspective agents is analogous, and order dependence in knowledge revision can emerge there as well---for instance, compare $(\cup_{j} K_j)^{\mathbf{d}\bm{\pm} \mathbf{d}}$ and $(\cup_{j} K_j^{\bm{\pm}\mathbf{d}})^{\mathbf{d}}$ in Table \ref{tab:ExKnow2}. Finally, while the seven knowledge operators in Proposition \ref{Prop:FeasRev} may differ when input knowledge is non-partitional, they all become equal when every input knowledge operator in $\mathbf{K}$ is partitional. This observation is formally stated and proved in Lemma \ref{lemma:K3} in Appendix \ref{Proof:FeasRev}.

\subsection{Revision types}\label{subsec:RevisionTypes}
Proposition \ref{Prop:FeasRev} identifies the possible outcomes of revising a given profile of knowledge operators. In Section \ref{Sec:DK}, we will show how agents attain distributed knowledge with a communication process in which they have to revise different profiles of knowledge operators at different moments. Hence we need to specify how agents revise their knowledge over time. We assume that everyone is consistent: While different agents may revise knowledge differently, any given agent revises in the same way all profiles of knowledge operators she might face. For example, a positively introspective agent $i$ always turns any feasible profile $(K_1, \dots, K_n)$ into $(K_1, \dots, K_n)^{\theta_i} = (\cup_{j} K_j)^{\bm{+} \mathbf{d}}$; another positively introspective agent, say $h$, always forms $(K_1, \dots, K_n)^{\theta_h} = (\cup_{j} K_j)^{\mathbf{d}{+}\mathbf{d}}$ out of any $(K_1, \dots, K_n)$. In compact form, we say that $i$'s revision type is $\theta_i = (\bm{+} \mathbf{d})$ while $h$'s revision type is $\theta_h = (\mathbf{d} \bm{+} \mathbf{d})$. Thus revision types are a concise representation of how agents revise any possible profile of knowledge operators they might deal with.

In light of Proposition \ref{Prop:FeasRev}, there are seven possible revision types, all of which are reported in Table \ref{tab:OrderRevision}: the first column indicates the outcome of knowledge revision for each type, whereas the second column indicates the corresponding label.

\begin{table}[htb]
	\centering
	\begin{tabular}{ c | c }
		$(K_1, \dots, K_n)^{\theta}$ &  Revision type $\theta$\\
		\hline			
	 $(\cup_{i} K_i)^{\mathbf{d}}$ & $(\mathbf{d})$ \\
	$(\cup_{i} K_i)^{\bm{+} \mathbf{d}}$ & $(\bm{+} \mathbf{d})$ \\
	$(\cup_{i} K_i)^{\mathbf{d}\bm{+} \mathbf{d}}$	& $(\mathbf{d} \bm{+} \mathbf{d})$\\
	$(\cup_{i} K_i^{\bm{+}\mathbf{d}})^{\mathbf{d}}$	& $(\bm{+} \mathbf{d}\text{s})$\\
	$(\cup_{i} K_i)^{\bm{\pm} \mathbf{d}}$	& $(\bm{\pm} \mathbf{d})$\\
	$(\cup_{i} K_i)^{\mathbf{d}\bm{\pm} \mathbf{d}}$	& $(\mathbf{d} \bm{\pm} \mathbf{d})$\\
	$(\cup_{i} K_i^{\bm{\pm}\mathbf{d}})^{\mathbf{d}}$	& $(\bm{\pm} \mathbf{d}\text{s})$
	\end{tabular}
	\caption{Revision types.}
	\label{tab:OrderRevision}
\end{table}

\subsection{Revision in terms of possibility relations}\label{subsec:CharacPossibilityRelations}
When an agent revises a profile $(K_1, \dots, K_n)$, she forms a knowledge operator $(K_1, \dots, K_n)^{\theta}$ that depends on her type $\theta \in \Theta$, where $\Theta$ is the set of all revision types. Since knowledge operators and possibility relations are dually isomorphic (Lemma \ref{lemma:Duality}), it is natural to look for the possibility relation representing $(K_1, \dots, K_n)^{\theta}$. The following proposition identifies the possibility relations representing the outcome of knowledge revision for each revision type.

\begin{proposition}\label{Prop:Representation}
Let $K_1, \dots, K_n$ be $n\geq 1$ knowledge operators in  $\mathcal{K}^1$, and let $\mathbf{P} = \left(P_{1}, \dots, P_{n}\right)$ be the profile of possibility relations in $\mathcal{P}^1$ representing those $n$ knowledge operators. Then the following hold.
\begin{itemize}
\item[(i)] $(\cup_{i} K_i)^{\mathbf{d}}$ is represented by $\cap_{i} P_{i}$;
\item[(ii)] $(\cup_{i} K_i)^{\bm{+}\mathbf{d}}$ is represented by $T_{\mathbf{P}}$, i.e., the left $n$-ary trace of $\mathbf{P}$;
\item[(iii)] $(\cup_{i} K_i)^{\mathbf{d}\bm{+}\mathbf{d}}$ is represented by $T_{\cap_{i} P_{i}}$, i.e., the left trace of $\cap_{i} P_{i}$;
\item[(iv)] $(\cup_{i} K_i^{\bm{+}\mathbf{d}})^{\mathbf{d}}$ is represented by $\cap_i T_{P_i}$;
\item[(v)] $(\cup_{i} K_i)^{\bm{\pm}\mathbf{d}}$ is represented by $E_{\mathbf{P}}$, i.e., the symmetric part of the left $n$-ary trace of $\mathbf{P}$;
\item[(vi)] $(\cup_{i} K_i)^{\mathbf{d}\bm{\pm}\mathbf{d}}$ is represented by $E_{\cap_{i} P_{i}}$, i.e., the symmetric part of the left trace of $\cap_{i} P_{i}$;
\item[(vii)] $(\cup_{i} K_i^{\bm{\pm}\mathbf{d}})^{\mathbf{d}}$ is represented by $\cap_i E_{P_i}$.
\end{itemize}
\end{proposition}
\begin{proof}
See Appendix \ref{Proof:Representation}.
\end{proof}

We know from Subsection \ref{subsec:Duality} that every knowledge operator satisfying distributivity and necessitation is represented by a unique possibility relation, which can be calculated via equation \eqref{eq:ReprMSZ}. The proposition above shows that the unique relation representing $(K_1, \dots, K_n)^{\theta}$ can be equivalently obtained by simple transformations of the possibility relations $\left(P_{1}, \dots, P_{n}\right)$, which in turn represent the input knowledge operators $(K_1, \dots, K_n)$. For example, the operator $(\cup_{i} K_i)^{\mathbf{d}}$, which would be formed by a non-introspective agent, is represented just by the intersection of the possibility relations $\left(P_{1}, \dots, P_{n}\right)$. On the other hand, the representations of the six remaining knowledge operators involve the left traces or the left $n$-ary trace of $\left(P_{1}, \dots, P_{n}\right)$, or their symmetric parts. Notice that the possibility relations in parts (ii)--(iv) are necessarily transitive and represent revised knowledge for positively introspective agents, whereas those in parts (v)--(vii) are all equivalence relations and represent revised knowledge for fully introspective agents. The content of Proposition \ref{Prop:Representation} is illustrated further in the example in Subsection \ref{subsec:Example1}.

\subsection{A partial order on the set of revision types}\label{subsec:PartialOrderTypes}
There is a natural order on the set of revision types. Let $\preceq$ be the binary relation on $\Theta$ defined as follows: For all types $\theta, \theta' \in \Theta$, we have $\theta \preceq \theta'$ if for all finite state spaces\footnote{Note that $\preceq$ is not necessarily a \textit{partial} order if one defines it relative to a fixed state space.} $\Omega$,
\begin{equation*}
(K_1, \dots, K_n)^{\theta} \leq (K_1, \dots, K_n)^{\theta'}
\end{equation*}
for all finite sequences of $n\geq 1$ operators $K_1, \dots, K_n$ in $\mathcal{K}^1(\Omega)$. In words, $\theta \preceq \theta'$ if revised knowledge for type $\theta'$ is always more detailed than revised knowledge for type $\theta$. It is easy to check that $\preceq$ is a partial order, i.e., a reflexive, transitive and antisymmetric binary relation. The Hasse diagram of $\preceq$ is represented in Figure \ref{fig:HassePreorder}, and it is primarily based on Lemma \ref{lemma:K1} in Appendix \ref{Proof:K1}.

\begin{figure}[h]
\centering
\begin{tikzpicture}[auto, node distance=2cm,
                    thick,main node/.style={}]

  \node[main node] (1) {$(\bm{\pm}\mathbf{d}\text{s})$};
  \node[main node] (3) [below  left of=1] {$(\bm{+}\mathbf{d}\text{s})$};
  \node[main node] (2) [below right of=1] {$(\bm{\pm}\mathbf{d})$};
  \node[main node] (4) [below right of=3] {$(\bm{+}\mathbf{d})$};
  \node[main node] (5) [below right of=4] {$(\mathbf{d}\bm{+}\mathbf{d})$};
  \node[main node] (6) [below right of=2] {$(\mathbf{d}\bm{\pm}\mathbf{d})$};
  \node[main node] (7) [below of=5] {$(\mathbf{d})$};

  \path[every node/.style={font=\small}]
    (2) edge node {} (1)
    (3) edge node {} (1)
    (4) edge node {} (3)
    (4) edge node {} (2)
    (5) edge node {} (4)
    (2) edge node {} (6)
    (5) edge node {} (6)
    (5) edge node {} (7);   
\end{tikzpicture}
	\caption{The partial order $\preceq$ on the set of revision types.}
\label{fig:HassePreorder}
\end{figure}
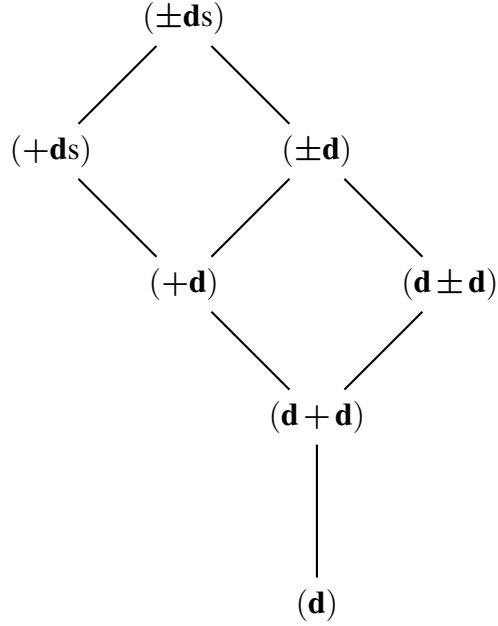

The pair $(\Theta, \preceq)$ forms a lattice. The top and bottom elements are $(\bm{\pm} \mathbf{d}\text{s})$ and $(\mathbf{d})$, respectively. There are types that are not comparable, e.g., $(\bm{+} \mathbf{d}\text{s})$ and $(\bm{\pm} \mathbf{d})$. This means that revised knowledge for $(\bm{\pm} \mathbf{d})$ is not always more detailed than revised knowledge for $(\bm{+} \mathbf{d}\text{s})$, and vice versa. In other words, each type cannot replicate all the inferences that the other type can conceivably draw.

\subsection{Example}\label{subsec:Example1}
This example refers to Propositions \ref{Prop:FeasRev} and \ref{Prop:Representation}. Suppose there are four possible states of the world. There are three agents whose knowledge is described by the operators $K_1, K_2$, and $K_3$ in Table \ref{tab:ExKnowOp1}.

\begin{table}[h]
	\centering
	\begin{tabular}{ c | c | c | c | c}
		$A$ &  $K_1(A)$  &  $K_2(A)$ &  $K_3(A)$  &  $\cup_i K_i(A)$\\
		\hline			
		$\emptyset$ & $\emptyset$ & $\emptyset$ & $\emptyset$ & $\emptyset$\\
		$\{\omega_1\}$ & $\emptyset$ & $\{\omega_1\}$ & $\{\omega_1\}$ & $\{\omega_1\}$ \\
		$\{\omega_2\}$ & $\{\omega_2\}$ & $\{\omega_2\}$ & $\emptyset$ & $\{\omega_2\}$\\
		$\{\omega_3\}$ & $\emptyset$ & $\emptyset$ & $\emptyset$ & $\emptyset$\\
		$\{\omega_4\}$ & $\emptyset$ & $\emptyset$ & $\emptyset$ & $\emptyset$\\
		$\{\omega_1, \omega_2\}$ & $\{\omega_1, \omega_2\}$ & $\{\omega_1, \omega_2\}$ & $\{\omega_1\}$ & $\{\omega_1, \omega_2\}$\\
		$\{\omega_1, \omega_3\}$ & $\emptyset$ & $\{\omega_1\}$ & $\{\omega_1\}$ & $\{\omega_1\}$\\
		$\{\omega_1, \omega_4\}$ & $\emptyset$ & $\{\omega_1\}$ & $\{\omega_1\}$ & $\{\omega_1\}$\\
		$\{\omega_2, \omega_3\}$ & $\{\omega_2\}$ & $\{\omega_2\}$ & $\emptyset$ & $\{\omega_2\}$\\
		$\{\omega_2, \omega_4\}$ & $\{\omega_2\}$ & $\{\omega_2\}$ & $\emptyset$ & $\{\omega_2\}$\\
		$\{\omega_3, \omega_4\}$ & $\{\omega_3\}$ & $\emptyset$ & $\emptyset$ & $\{\omega_3\}$\\
		$\{\omega_1, \omega_2, \omega_3\}$ & $\{\omega_1, \omega_2\}$ & $\{\omega_1, \omega_2\}$ & $\{\omega_1\}$ & $\{\omega_1, \omega_2\}$\\
		$\{\omega_1, \omega_2, \omega_4\}$ & $\{\omega_1, \omega_2\}$ & $\{\omega_1, \omega_2\}$ & $\{\omega_1\}$ & $\{\omega_1, \omega_2\}$\\
		$\{\omega_1, \omega_3, \omega_4\}$ & $\{\omega_3, \omega_4\}$ & $\{\omega_1\}$ & $\{\omega_1\}$ & $\{\omega_1, \omega_3, \omega_4\}$\\
		$\{\omega_2, \omega_3, \omega_4\}$ & $\{\omega_2, \omega_3\}$ & $\{\omega_2, \omega_3, \omega_4\}$ & $\{\omega_2, \omega_3, \omega_4\}$ & $\{\omega_2, \omega_3, \omega_4\}$\\
		$\Omega$ & $\Omega$ & $\Omega$ & $\Omega$ & $\Omega$
	\end{tabular}
	\caption{Initial knowledge operators and their union.}
	\label{tab:ExKnowOp1}
\end{table}

The operators $K_1, K_2$, and $K_3$ are represented by the possibility relations $P_1, P_2$, and $P_3$ in Table \ref{tab:ExPosRel}, respectively. Notice that $K_3$ is partitional, i.e., it is represented by the equivalence relation $P_3$, whereas $K_1$ and $K_2$ are not.

\begin{table}[h]
	\centering
	\begin{tabular}{ c | c | c | c | c | c | c}
		$\omega$ &  $P_1(\omega)$  &  $P_2(\omega)$ &  $P_3(\omega)$ &  $\begin{matrix*}[l] \cap_i P_i(\omega), \\ T_{\cap_i P_i}(\omega),\\ E_{\cap_i P_i}(\omega)\end{matrix*}$ &  $\begin{matrix*}[l] T_{\mathbf{P}}(\omega), \\ \cap_i T_{P_i}(\omega)\end{matrix*}$ &  $\begin{matrix*}[l] E_{\mathbf{P}}(\omega), \\ \cap_i E_{P_i}(\omega)\end{matrix*}$ \\
		\hline			
		$\omega_1$ & $\{\omega_1,\omega_2\}$ & $\{\omega_1\}$ & $\{\omega_1\}$ & $\{\omega_1\}$ & $\{\omega_1\}$ & $\{\omega_1\}$\\
		$\omega_2$ & $\{\omega_2\}$ & $\{\omega_2\}$ & $\{\omega_2, \omega_3, \omega_4\}$ & $\{\omega_2\}$ & $\{\omega_2\}$ & $\{\omega_2\}$\\
		$\omega_3$ & $\{\omega_3, \omega_4\}$ & $\{\omega_2, \omega_3, \omega_4\}$ & $\{\omega_2, \omega_3, \omega_4\}$& $\{\omega_3, \omega_4\}$ & $\{\omega_3\}$ & $\{\omega_3\}$\\
		$\omega_4$ & $\{\omega_1, \omega_3, \omega_4\}$ & $\{\omega_2, \omega_3, \omega_4\}$ & $\{\omega_2, \omega_3, \omega_4\}$& $\{\omega_3, \omega_4\}$ & $\{\omega_3, \omega_4\}$ & $\{\omega_4\}$
	\end{tabular}
	\caption{Initial possibility relations, their intersection, and their traces.}
	\label{tab:ExPosRel}
\end{table}

If an agent learns the knowledge operators of the other two people, she will form a revised knowledge operator out of $(K_1, K_2, K_3)$ according to her own revision type. For instance, if the second agent's revision type is $\theta_2 =  (\bm{+}\mathbf{d})$, her revised knowledge operator is $(\cup_i K_i)^{\bm{+}\mathbf{d}}$. The latter can be found in Table \ref{tab:ExKnow2} together with the operators that, according to Proposition \ref{Prop:FeasRev}, represent revised knowledge for the other six possible revision types.

\begin{table}[h]
\centering
\begin{tabular}{ c | c | c | c}
$A$   &  $\begin{matrix*}[l] (\cup_i K_i)^{\mathbf{d}}(A), \\ (\cup_i K_i)^{\mathbf{d}\bm{+}\mathbf{d}}(A),\\ (\cup_i K_i)^{\mathbf{d}\bm{\pm}\mathbf{d}}(A)\end{matrix*}$&  $\begin{matrix*}[l] (\cup_i K_i)^{\bm{+}\mathbf{d}}(A), \\ (\cup_i K_i^{\bm{+}\mathbf{d}})^{\mathbf{d}} (A)\end{matrix*}$ &  $\begin{matrix*}[l] (\cup_i K_i)^{\bm{\pm}\mathbf{d}}(A), \\ (\cup_i K_i^{\bm{\pm}\mathbf{d}})^{\mathbf{d}} (A)\end{matrix*}$\\
  \hline			
  $\emptyset$  & $\emptyset$ & $\emptyset$ & $\emptyset$\\
  $\{\omega_1\}$   & $\{\omega_1\}$& $\{\omega_1\}$  & $\{\omega_1\}$\\
  $\{\omega_2\}$  & $\{\omega_2\}$ & $\{\omega_2\}$ & $\{\omega_2\}$\\
  $\{\omega_3\}$  & $\emptyset$ & $\{\omega_3\}$ & $\{\omega_3\}$\\
 $\{\omega_4\}$  & $\emptyset$ & $\emptyset$  & $\{\omega_4\}$\\
  $\{\omega_1, \omega_2\}$  & $\{\omega_1, \omega_2\}$ & $\{\omega_1, \omega_2\}$ &  $\{\omega_1, \omega_2\}$\\
  $\{\omega_1, \omega_3\}$  & $\{\omega_1\}$ & $\{\omega_1, \omega_3\}$ & $\{\omega_1, \omega_3\}$\\
   $\{\omega_1, \omega_4\}$  & $\{\omega_1\}$ & $\{\omega_1\}$ & $\{\omega_1, \omega_4\}$\\
   $\{\omega_2, \omega_3\}$  & $\{\omega_2\}$ & $\{\omega_2, \omega_3\}$ & $\{\omega_2, \omega_3\}$\\
   $\{\omega_2, \omega_4\}$  & $\{\omega_2\}$ & $\{\omega_2\}$ & $\{\omega_2, \omega_4\}$\\
   $\{\omega_3, \omega_4\}$  & $\{\omega_3, \omega_4\}$ & $\{\omega_3, \omega_4\}$ & $\{\omega_3, \omega_4\}$\\
   $\{\omega_1, \omega_2, \omega_3\}$  & $\{\omega_1, \omega_2\}$ & $\{\omega_1, \omega_2, \omega_3\}$ & $\{\omega_1, \omega_2, \omega_3\}$\\
   $\{\omega_1, \omega_2, \omega_4\}$  & $\{\omega_1, \omega_2\}$ & $\{\omega_1, \omega_2\}$ &  $\{\omega_1, \omega_2, \omega_4\}$\\
   $\{\omega_1, \omega_3, \omega_4\}$  & $\{\omega_1, \omega_3, \omega_4\}$ & $\{\omega_1, \omega_3, \omega_4\}$ & $\{\omega_1, \omega_3, \omega_4\}$\\
$\{\omega_2, \omega_3, \omega_4\}$  & $\{\omega_2, \omega_3, \omega_4\}$ & $\{\omega_2, \omega_3, \omega_4\}$ & $\{\omega_2, \omega_3, \omega_4\}$\\
$\Omega$  & $\Omega$ & $\Omega$  &$\Omega$
\end{tabular}
\caption{Revised knowledge operators.}
\label{tab:ExKnow2}
\end{table}

One can form the revised knowledge operators in Table \ref{tab:ExKnow2} by using the definitions of the relevant revision operators. For example, one can proceed as follows to calculate $(\cup_i K_i)^{\bm{+}\mathbf{d}}$. First, form the union $\cup_i K_i$, which can also be found in Table \ref{tab:ExKnowOp1}; next, form $(\cup_i K_i)^{\bm{+}}$ just by applying the definition of the positive introspection operator to $\cup_i K_i$; finally, apply the distributive closure to $(\cup_i K_i)^{\bm{+}}$ in order to obtain $(\cup_i K_i)^{\bm{+}\mathbf{d}}$. The process just described is straightforward yet laborious, since each step requires to calculate a new knowledge operator. We can cut down on calculations by using possibility relations. Specifically, according to Proposition \ref{Prop:Representation}, the knowledge operator $(\cup_i K_i)^{\bm{+}\mathbf{d}}$ is represented by $T_{\mathbf{P}}$, which is the left $n$-ary trace of $\mathbf{P} = (P_1, P_2, P_3)$. The latter can be found in Table \ref{tab:ExPosRel}, and it is obtained just by following the definition of the left $n$-ary trace in Subsection \ref{subsec:PRelations}. Next, we can use equation \eqref{eq:ReprKbyP} from Subsection \ref{subsec:Duality} to construct the unique knowledge operator represented by $T_{\mathbf{P}}$, which is exactly $(\cup_i K_i)^{\bm{+}\mathbf{d}}$. In Table \ref{tab:ExPosRel} we can also find the possibility relations that represent the remaining six revised knowledge operators of this example. To conclude, note that the equivalence relation $E_{\mathbf{P}}$ can also be formed through minimal sets as per Lemma \ref{lemma:MINsets} in Appendix \ref{MinimalSets}.

\section{Distributed knowledge}\label{Sec:DK}
This section gives a formal definition of distributed knowledge and then shows  how distributed knowledge is determined by the join (least upper bound) of the revision types in the group.

We define distributed knowledge as the knowledge that group members can attain by communicating with one another and sharing everything they know. Communication is non-strategic and truthful. Agents can communicate for as long as they want. Formally, time is discrete and indexed by $t = 0, 1, 2, \dots$. Let $(\theta_1, \dots, \theta_N)$ be the group's revision types, which are held fixed throughout the entire communication process. Let $(K_1, \dots, K_N)$ be a profile of initial knowledge operators in $\mathcal{K}^1$. At each stage $t$, everyone communicates her knowledge operator to everyone else. Then every agent revises her knowledge in light of what she has learned from others and according to her own revision type. At stage $t + 1$, the same process of communication and revision takes place again, and so on in all subsequent stages. The resulting knowledge of each agent $i$ is represented by the sequence $(K_i^t: t = 0, 1, 2,\dots)$ defined recursively as follows. For each $i\in I$,
\begin{align*}
K_i^0 &= K_i \\
K_i^t &= (K_1^{t-1}, \dots, K_N^{t-1})^{\theta_i} \quad \text{ for } t=1, 2, \dots
\end{align*}

Distributed knowledge is what anyone can know ``at the end'' of the communication process just described. Formally, distributed knowledge is the knowledge operator $K_{\mathbf{D}}$ obtained as
\begin{equation}\label{eq:DKdef}
K_{\mathbf{D}} = \cup_{i\in I}\left(\cup_{t = 0}^{\infty} K_i^t\right).
\end{equation}

A couple of remarks. First, agents are not required to know the other group members' revision types. The reason is that the way in which an agent revises knowledge depends solely on her own revision type and not on others'. Second, while we assume that all initial knowledge operators lie in $\mathcal{K}^1$, we could alternatively assume that initial operators reflect the agents' level of introspection. That is, we could assume that an initial knowledge operator is in $\mathcal{K}^2$ or $\mathcal{K}^3$ if the agent under consideration is positively or fully introspective, respectively. This alternative assumption is slightly less general than the one we actually make and, more substantially, would not alter any result of this or the preceding section.

The next theorem shows how distributed knowledge in \eqref{eq:DKdef} reduces to a simple expression that depends only on the initial knowledge operators and the join of the revision types in the group. 

\begin{theorem*}
Suppose $(\theta_1, \dots, \theta_N)$ are the revision types of agents in $I$. Call $\bar{\theta}$ the corresponding join, i.e., $\bar{\theta}:= \theta_1 \vee \dots \vee \theta_N$. In addition, suppose $(K_1, \dots, K_N)$ are the initial knowledge operators, with $K_i \in \mathcal{K}^1$ for all $i \in I$. Then distributed knowledge in $I$ is
\begin{equation*}
	K_{\mathbf{D}} = (K_1, \dots, K_N)^{\bar{\theta}}.
\end{equation*}
\end{theorem*}

\begin{proof}
See Appendix \ref{Proof:THM}.
\end{proof}

The theorem says that distributed knowledge can be found by revising the initial knowledge operators according to the join of the revision types in the group. The theorem covers two distinct cases. In the first, there is a group member, say $N$, whose revision type is equal to the join $\bar{\theta}$. This means that $\theta_i \preceq \theta_N$ for all $i\in I$; in words, agent $N$ can make all the inferences that anyone else in the group makes. As a consequence, distributed knowledge coincides with what agent $N$ knows at the end of the first stage of communication, in which $N$ learns everyone else's knowledge operator. This result is perfectly in line with the common view that distributed knowledge is just the knowledge of one distinct group member, to whom all other agents communicate what they know. Notice that in this case one round of unilateral communication is enough to attain distributed knowledge.

In the second case, no one in the group has a revision type equal to the join $\bar{\theta}$. Hence, there are two group members, say $N - 1$ and $N$, whose types are maximal and incomparable. Further, as can be seen from the Hasse diagram in Figure \ref{fig:HassePreorder}, one maximal type, say $N - 1$, must be positively introspective while the other must be fully introspective. Since the two maximal types are incomparable, it may well be the case that $K_{N-1}^1 \leq K_N^1$ does not hold. Intuitively, the positively introspective maximal type might make certain inferences that are not made by the fully introspective one. Nonetheless, as we show rigorously in the theorem's proof, at the end of the second communication stage we have that $K_{N-1}^2 \leq K_{N}^2$, i.e., the fully introspective maximal type has more detailed knowledge than the other maximal type has. This happens because, roughly speaking, while $N$'s revised knowledge $K_{N}^1$ does not incorporate the inferences made by the positively introspective agent $N - 1$ at the end of the first stage of communication, $N$'s revised knowledge $K_{N}^2$ \textit{does} contain those inferences. As a result, after two rounds of communication the fully introspective $N$ has the most detailed knowledge in the group, and no other group member can strictly expand the partitional operator $K_N^2$ in subsequent communication stages. Therefore, distributed knowledge coincides with what agent $N$ knows at the end of the second, but not necessarily the first, stage of communication. This result is at odds with the interpretation of distributed knowledge as the knowledge of a distinct agent to whom all other group members communicate their knowledge. In fact, even though distributed knowledge coincides with what agent $N$ knows at the end of the second communication stage, $N$ cannot form $K_N^2$ just by learning the others' initial knowledge operators. A preceding communication stage is necessary, in which the other maximal type forms $K_{N-1}^1$ out of the initial knowledge operators. Once agent $N$ learns $K_{N-1}^1$ in the second stage, she can merge it with her previous knowledge $K_{N}^1$ to form $K_N^2$, which coincides with distributed knowledge. In sum, distributed knowledge is the combined knowledge of two distinct agents, to whom everyone communicates her initial knowledge.

The case just discussed can be interpreted as a wisdom-of-crowd effect. The Theorem implies that a group where no one has type $\bar{\theta}$ attains the same distributed knowledge that would have been attained if $\bar{\theta}$ had been a group member's type. However, it takes longer for a group without $\bar{\theta}$ to attain distributed knowledge. Relatedly, distributed knowledge can be strictly more detailed than what each group member would come to know just by knowing everyone's initial knowledge. It is worth noting that the wisdom-of-crowd effect can emerge only in groups of agents having different levels of introspection. When all group members are equally introspective, the join of their revision types is already present in the group.

The example in Subsection \ref{subsec:Example2} illustrates the Theorem's main results, including the wisdom-of-crowd effect. To conclude, the following corollary shows how we can simplify the Theorem when knowledge is partitional, as it is often assumed in economics and game theory.

\begin{corollary*}
Suppose all the initial knowledge operators are represented by equivalence relations, i.e., $K_i \in \mathcal{K}^3$ for all $i \in I$. Then for any profile of revision types $(\theta_1, \dots, \theta_N)$, distributed knowledge in $I$ is
\begin{equation}\label{eq:DKpart}
K_{\mathbf{D}} = (K_1, \dots, K_N)^{\mathbf{d}}.
\end{equation}
\end{corollary*}
\begin{proof}
The result follows easily from Lemma \ref{lemma:K3} in Appendix \ref{Proof:FeasRev} and the theorem above.
\end{proof}

When knowledge is partitional, distributed knowledge can be found just by applying the distributive closure to the union of the initial knowledge operators. The distribution of revision types does not affect distributed knowledge. To see why this is the case, notice that $K_{\mathbf{D}}$ in \eqref{eq:DKpart} is represented by the intersection of the equivalence relations representing $K_1, \dots, K_N$. Since the intersection of equivalence relations is still an equivalence relation, the operator $K_{\mathbf{D}}$ is partitional, i.e., $K_{\mathbf{D}} \in \mathcal{K}^3$. Hence, no revision operator can strictly expand $K_{\mathbf{D}}$.

\subsection{Example (continued)}\label{subsec:Example2}
The setting is the same as in the example in Subsection \ref{subsec:Example1}, with three agents whose initial knowledge operators are $K_1, K_2$, and $K_3$ in Table \ref{tab:ExKnowOp1}. Our goal here is to illustrate the content of the Theorem above. We are going to consider two different profiles of revision types, reflecting the two cases covered by the Theorem. The detailed description of every revised knowledge operator mentioned in this Subsection can be found in Table \ref{tab:ExKnow2}, whereas the possibility relations that represent them can be found in Table \ref{tab:ExPosRel}.

To begin with, suppose the agents' revision types are $(\theta_1, \theta_2, \theta_3) = ((\mathbf{d}), (\mathbf{d}\bm{+}\mathbf{d}), (\mathbf{d}\bm{\pm}\mathbf{d}))$, so that the first agent is non-introspective, the second is positively introspective, and the third is fully introspective. As can be seen from the Hasse diagram in Figure \ref{fig:HassePreorder}, the join of the three revision types is $\bar{\theta} = (\mathbf{d}\bm{\pm}\mathbf{d}) = \theta_3$. By the Theorem above, distributed knowledge is
\begin{equation*}
K_{\mathbf{D}} = (K_1, K_2, K_3)^{\bar{\theta}} = (K_1, K_2, K_3)^{(\mathbf{d}\bm{\pm}\mathbf{d})} = (\cup_i K_i)^{\mathbf{d}\bm{\pm}\mathbf{d}}.
\end{equation*}

Since agent $3$'s revision type is also the join, distributed knowledge can be attained in one round of communication just by transferring each agent's initial knowledge to agent $3$. That is, if agents $1$ and $2$ share their initial knowledge operators with agent $3$, the latter will revise $(K_1, K_2, K_3)$ according to her own revision type and form $(\cup_i K_i)^{\mathbf{d}\bm{\pm}\mathbf{d}}$, which is exactly distributed knowledge. We have already discussed in Subsection \ref{subsec:Example1} how to form a revised knowledge operator by using the possibility relation that represents it. By Proposition \ref{Prop:Representation}, here we have to use the equivalence relation $E_{\cap_i P_i}$ to form $(\cup_i K_i)^{\mathbf{d}\bm{\pm}\mathbf{d}}$. In sum, when revision types are $(\theta_1, \theta_2, \theta_3) = ((\mathbf{d}), (\mathbf{d}\bm{+}\mathbf{d}), (\mathbf{d}\bm{\pm}\mathbf{d}))$, distributed knowledge is just the knowledge of one distinct group member, namely agent $3$, to whom all other agents communicate what they know.

Now suppose that revision types are $(\theta_1', \theta_2', \theta_3') = ((\mathbf{d}), (\bm{+}\mathbf{d}), (\mathbf{d}\bm{\pm}\mathbf{d}))$. We can see from the Hasse diagram in Figure \ref{fig:HassePreorder} that the join is $(\bm{\pm}\mathbf{d})$, which is not a type in the group. The two maximal types are $\theta_2'$ and $\theta_3'$. By the Theorem above, distributed knowledge is
\begin{equation*}
K_{\mathbf{D}} = (K_1, K_2, K_3)^{\bar{\theta}} = (K_1, K_2, K_3)^{(\bm{\pm}\mathbf{d})} = (\cup_i K_i)^{\bm{\pm}\mathbf{d}}.
\end{equation*}

This is a case in which two rounds of communication are necessary to attain distributed knowledge. Specifically, at the end of the first communication stage, each agent learns the other agents' knowledge operators and forms a revised knowledge operator according to her own revision type. The revised operators at the end of the first stage are
\begin{align*}
K_1^1 & = (\cup_i K_i)^{\mathbf{d}},\\
K_2^1 & = (\cup_i K_i)^{\bm{+}\mathbf{d}},\\
K_3^1 &= (\cup_i K_i)^{\mathbf{d}\bm{\pm}\mathbf{d}}.
\end{align*}

As can be seen from Table \ref{tab:ExKnow2}, we have $(\cup_i K_i)^{\mathbf{d}} = (\cup_i K_i)^{\mathbf{d}\bm{\pm}\mathbf{d}}$ and $(\cup_i K_i)^{\bm{+}\mathbf{d}} \not\leq (\cup_i K_i)^{\mathbf{d}\bm{\pm}\mathbf{d}}$. That is, agents $1$ and $3$ do not have more detailed knowledge than agent $2$ has. In fact, agent $2$'s revised knowledge $K_2^1$ is strictly more detailed than that of the other two group members. Intuitively, agent $2$ revises $(K_1, K_2, K_3)$ by making certain inferences that are not made by the other agents. Now, in the second round of communication, agents share their revised knowledge $(K_1^1, K_2^1, K_3^1)$ and then revise it again. As a result, the knowledge operators at the end of the second communication stage are
\begin{align*}
K_1^2 & = (\cup_i K_i^1)^{\mathbf{d}} = \left((\cup_i K_i)^{\bm{+}\mathbf{d}}\right)^{\mathbf{d}} = (\cup_i K_i)^{\bm{+}\mathbf{d}},\\
K_2^2 & = (\cup_i K_i^1)^{\bm{+}\mathbf{d}} = \left((\cup_i K_i)^{\bm{+}\mathbf{d}}\right)^{\bm{+}\mathbf{d}} = (\cup_i K_i)^{\bm{+}\mathbf{d}},\\
K_3^2 & = (\cup_i K_i^1)^{\mathbf{d}\bm{\pm}\mathbf{d}} = \left((\cup_i K_i)^{\bm{+}\mathbf{d}}\right)^{\mathbf{d}\bm{\pm}\mathbf{d}} = (\cup_i K_i)^{\bm{\pm}\mathbf{d}}.
\end{align*}

By inspecting again Table \ref{tab:ExKnow2}, we can see that $(\cup_i K_i)^{\bm{+}\mathbf{d}} \leq (\cup_i K_i)^{\bm{\pm}\mathbf{d}}$ and $(\cup_i K_i)^{\bm{+}\mathbf{d}} \neq (\cup_i K_i)^{\bm{\pm}\mathbf{d}}$. Hence, now agent $3$ has the most detailed knowledge in the group. What is more, distributed knowledge coincides with what agent $3$ knows at the end of the second communication stage, i.e., $K_{\mathbf{D}} = K_3^2 = (\cup_i K_i)^{\bm{\pm}\mathbf{d}}$. This is true because, after the second round of communication, the union of the individual knowledge operators is exactly $K_3^2$; since the latter is partitional, it cannot be strictly expanded by any revision operator in subsequent communication rounds.

In sum, the case with revision types $(\theta_1', \theta_2', \theta_3') = ((\mathbf{d}), (\bm{+}\mathbf{d}), (\mathbf{d}\bm{\pm}\mathbf{d}))$ has three distinguishing features. First, two rounds of communication are necessary to attain distributed knowledge. Second, distributed knowledge can no longer be interpreted as knowledge of one distinct group member, to whom all other agents communicate just their initial knowledge. Differently put, distributed knowledge is strictly more detailed than what any group member knows at the end of the first round of communication. Finally, a wisdom-of-crowd effect emerges: although the join $(\bm{\pm}\mathbf{d})$ is not a revision type in the group, the three agents attain the same distributed knowledge that they would have attained if one of their revision types had been the join $(\bm{\pm}\mathbf{d})$.

\section{Discussion}\label{Sec:Disc}
\subsection{Related work}
Our paper connects the literature on set-theoretic models of distributed knowledge to that on rational dialogues and consensus. In essence, we formalize distributed knowledge as knowledge attainable through a specific type of rational dialogue \textit{\`{a} la} \cite{geanakoplos1982}.

Initiated by \cite{aumann1976agreeing}, the set-theoretic (or semantic) approach to interactive knowledge is standard practice in economics and game theory (see \cite{aumann1999interactive1} for an overview). In contrast to our paper, knowledge is almost invariably partitional in set-theoretic models. Against this backdrop, distributed knowledge has been employed especially in the theory of cooperative games: examples include the fine core \citep{wilson1978informationCore} and information sharing games \citep{slikker2003information-sharing, sasaki2025core}. The lattice-theoretic properties of distributed knowledge in partitional models are studied by \cite{tobias2021meet}. 

\cite{Geanakoplos1989NonPart, GeanakoplosNonPartition} introduces set-theoretic models with non-partitional knowledge and examines their implications for common knowledge, decision theory, game theory, and no-trade theorems; \cite{brandenburger-d-g1992correlated} adopt a similar non-partitional model of information to extend the notion of correlated equilibrium. Our paper builds on \cite{Geanakoplos1989NonPart, GeanakoplosNonPartition} and extends his analysis to distributed knowledge. \cite{fukuda2019distributed} is very close to our work, too. Dispensing with the veridicality axiom, he allows agents to hold false beliefs. In this regard, his analysis is more general than ours. On the other hand, \cite{fukuda2019distributed} does not model the communication process that leads to distributed knowledge. We provide such a model and show that the inference making that comes with communication is crucial in determining what is distributed knowledge in a group. In essence, \cite{fukuda2019distributed} formulates distributed knowledge as knowledge that can be deduced from collective information, whereas our notion of distributed knowledge is based on communication. \cite{Tallon-Vergnaud-Zamir} study communication and belief revision in a semantic model similar to ours, yet they do not examine distributed knowledge. They work in the modal system \textit{KD45}, so allowing for mistaken beliefs. They show, among other things, that belief revision depends on the order in which communication takes place. The order dependence in their model is driven by mistaken beliefs. By contrast, the order dependence of knowledge revision in our model is caused by the lack of positive or negative introspection.

Outside of set-theoretic models, an early study of distributed knowledge is \cite{hilpinen1977remarks}, in philosophy. Later, distributed knowledge has been extensively studied in theoretical computer science and epistemic logic, with \cite{HalpernMoses-Distributed} being most likely the first to offer a formal model in terms of Kripke structures. Subsequent works include, among others, \cite{derHoek-et-alDistributed}, \cite{Roelofse-distributed}, \cite{aagotnes2017resolving}, \cite{BaltagSmets_Learning_What_Others_Know} and \cite{balbiani-ditmarsch_DDK}. The communication protocol in our paper is analogous to ``fully public'' information sharing in \cite{BaltagSmets_Learning_What_Others_Know}. There are two crucial differences between our paper and those from theoretical computer science and epistemic logic. First, while computer scientists and logicians study knowledge using Kripke structures, we use Aumann structures instead. Second, while distributed knowledge is commonly studied for a fixed modal system, we cover cases where different agents in the same group satisfy different modal systems. Specifically, our three categories of agents---namely fully introspective, positively introspective and non-introspective---correspond to the three modal systems \textit{S5}, \textit{S4} and \textit{T}, respectively.

Agents in our model attain distributed knowledge through communication in much the same way as agents in the literature on rational dialogues attain common knowledge and consensus. The literature on rational dialogues was initiated by \cite{geanakoplos1982}; subsequent works include, among many others, \cite{parikh-krasucki1990} and \cite{krasucki1996}. Our model is essentially a rational dialogue with three distinguishing features. First, in a typical rational dialogue agents communicate a fraction of their private information, whereas in our model agents communicate everything they know. Second, knowledge revision in our model is carried out in terms of knowledge operators; in the literature on rational dialogues, knowledge revision is done in terms of information partitions. Finally, to the best of our knowledge, the literature on rational dialogues has focused exclusively on fully introspective and rational agents, whereas we take into account that agents may have imperfect introspection skills. In spite of such differences, when knowledge is partitional, the distributed knowledge introduced in this paper coincides with the knowledge that agents have when they reach a consensus in the rational dialogue studied in \cite{crescenzi2022learning}, provided that the message function used during the dialogue reveals everyone's private information \citep[Subsection 5.6]{crescenzi2022learning}.

The form of bounded rationality studied in this paper is an instance of what \cite{grabiszewski2015epistemic_EB} calls ``epistemic bounded rationality,'' in which positive or negative introspection do not necessarily hold. The role that reasoning and introspective skills play in models of interactive knowledge and beliefs is explored by \cite{gossner-tsakas2012reasoning}, \cite{grabiszewski2015epistemic_EB}, and most recently \cite{fukuda2025meta-certain}. None of these papers study distributed knowledge, though.

Revision operators in our model bear some resemblance to the hypothetical knowledge operators introduced by \cite{samet1996hypothetical}.\footnote{I thank an anonymous reviewer for suggesting the connection between revision operators and hypothetical knowledge operators.} Roughly speaking, a hypothetical knowledge operator describes what an agent would have known if a certain hypothesis had been true, where a hypothesis can be any non-empty event. On the other hand, a revision operator captures what knowledge an agent can derive only from epistemic events that describe her knowledge. In addition, while hypothetical knowledge operators enrich standard models maintaining the assumption of partitional knowledge, there is a sense in which revision operators require non-partitional knowledge in order to bear fruit. In fact, the revision operators introduced in this paper cannot strictly expand a partitional knowledge operator.

Our work is also related to the AGM belief revision theory of \cite{alchourron1985logic}. Like in their framework, we model how agents revise knowledge upon receiving new information. There is a crucial difference, though. In AGM theory, new information is in the form of a single proposition; in our paper, new information is given by a knowledge operator that describes what an agent knows at every possible state. This means that, in our model, agents revise both their actual and their counterfactual knowledge. The learning process in our model is also analogous to \cite{basu2019bayesian} and \cite{sadler2021practical}. While \cite{basu2019bayesian} studies probabilistic beliefs in an AGM framework, we study non-probabilistic beliefs and allow for revision of counterfactual knowledge. As for \cite{sadler2021practical}, his model is set-theoretic as it is ours. Nevertheless, he does not consider counterfactual knowledge and his learning axioms are rooted in behavioral economics, whereas our revision operators are based on the standard introspection axioms from the literature on interactive knowledge and beliefs.

Finally, some of our results are analogous to those of \cite{mueller2014oneBayesian}. In the context of observational learning, he studies belief dynamics in groups with both Bayesian and non-Bayesian agents. His main result is that belief aggregation is driven by the Bayesian members of the group---that is, by the ``most rational'' group members. This is conceptually close to our Theorem in Section \ref{Sec:DK}, in which we show that distributed knowledge is determined by the join of the revision types in the group.

\subsection{Knowledge operators vs. possibility relations}
The equivalence between knowledge operators and possibility relations requires some qualification. Lemma \ref{lemma:Duality} says that knowledge operators and possibility relations are equivalent if knowledge satisfies necessitation and distributivity. But distributivity may fail to hold at certain steps of knowledge revision. Take, for example, the operator $K$ in Table \ref{tab:PosInt}. When we apply the positive introspection operator to it, we obtain the new operator $K^{\bm{+}}$, which violates distributivity and cannot be represented by any possibility relation. The ``closest'' knowledge operator representable by a relation is $K^{\bm{+} \mathbf{d}}$, which is represented by the left trace of the possibility relation representing $K$. Thus, knowledge operators give us a complete description of the inferences made at each step of knowledge revision, whereas possibility relations do not.

\subsection{The negative introspection operator}
In Section \ref{sec:RevisionOperators} we do not consider a revision operator that captures only negative introspection. The reason is that such an operator would eventually coincide with the full introspection operator. To see why this is the case, suppose $(.)^{\bm{-}}$ is the revision operator that maps each $K$ in $\mathcal{K}^v$ to the knowledge operator $K^{\bm{-}}$ constructed as follows. For all $A \subseteq \Omega$,
\begin{equation*}
	K^{\bm{-}} (A) = 
	\begin{cases}
		A & \text{ if } A \in \mathsf{Img}(\lnot K)\\
		K(A) & \text{otherwise}.
	\end{cases}
\end{equation*}
One can check that (a) the operator $(.)^{\bm{-}}$ is not idempotent and (b) applying $(.)^{\bm{-}}$ twice gives exactly the full introspection operator, i.e., $K^{\bm{--}} = K^{\bm{\pm}}$.

\subsection{On the notion of distributed knowledge}
We mentioned in the Introduction that distributed knowledge can be interpreted as knowledge of a ``wise man'' who has access to what each group member knows. The wise man is usually interpreted as an external observer. This interpretation raises questions about the wise man's reasoning skills when the group consists of differently introspective agents. In particular, what is the wise man's revision type? Should it be the join of the revision types in the group? Or should we interpret the wise man as someone having the highest feasible type $(\bm{\pm} \mathbf{d}\text{s})$ regardless of the actual distribution of types in the group? Taking a more pragmatic approach, our formulation circumvents these conceptual issues. We define distributed knowledge as what group members actually come to know through communication and full information sharing, without postulating any hypothetical individual outside the group. Our formulation is close to the notion of ``resolving distributed knowledge'' of \cite{aagotnes2017resolving}; it is also close to the principle of full communication put forth by \cite{derHoek-et-alDistributed} and further examined by \cite{Roelofse-distributed}. The principle of full communication says that if an event is distributed knowledge, then it should be possible for group members to establish that event through communication. We follow the principle by defining distributed knowledge exactly as what agents can establish through communication.

Another aspect of our formalization is the interpretation of epistemic events that are distributed knowledge. It is well known that the standard semantics of distributed knowledge may give rise to perplexing interpretations---see \cite{aagotnes2017resolving}, among others. To see this, consider our motivating example. Before any communication, the event ``Carol knows that $\{\omega_2\}$'' is the empty set. After communication, the event ``Carol knows that $\{\omega_2\}$'' is equal to $\{\omega_2\}$. Is the event ``Carol knows that $\{\omega_2\}$'' distributed knowledge in $\omega_2$? If we take this event to be the empty set, it is clear that it cannot be distributed knowledge. But we believe that the following interpretation is more appropriate. After talking with Bob, Carol comes to know that $\{\omega_2\}$ holds in $\omega_2$. Thus ``Carol knows that $\{\omega_2\}$'' is equal to $\{\omega_2\}$ and is distributed knowledge in $\omega_2$. More generally, since we define distributed knowledge as knowledge attained through communication, we implicitly assume that epistemic events are formed and evaluated with respect to what agents know after communication, which is represented by $K_{\mathbf{D}}$ for everyone.\footnote{Correspondence with Hans van Ditmarsch, whom I thank, helped me clarify this point.} In the epistemic logic literature, \cite{balbiani-ditmarsch_DDK} take a similar approach.

\subsection{Common knowledge}
What is the relationship between common and distributed knowledge? In essence, common knowledge is a static concept: it refers to events that are public at a given point in time---see \citet[p. 385]{GeanakoplosNonPartition}. As such, common knowledge does not presuppose any communication process. It is represented by the transitive closure of the union of the agents' possibility relations, and does not depend on the agents' revision types. Distributed knowledge is a dynamic concept: it describes what agents come to know by sharing all their private information. If an event is common knowledge, then it is also distributed knowledge. The converse need not be true. However, when people communicate as they do in Section \ref{Sec:DK}, they eventually reach a situation where common and distributed knowledge coincide in that everybody has the same knowledge. Because of this, we can also interpret our notion of distributed knowledge as \textit{potential} common knowledge: distributed knowledge is what becomes common knowledge through communication.

\newpage

\newpage
\appendix
\section{Appendix}
\subsection{A characterization of the symmetric part $E_{\mathbf{P}}$}\label{MinimalSets}
Here we give a characterization of $E_{\mathbf{P}}$ in terms of minimal sets. This is helpful because it enables us to find $E_{\mathbf{P}}$ without computing the left $n$-ary trace $T_{\mathbf{P}}$. Given $\mathbf{P}$, the minimal set of $\mathbf{P}$ at $\omega$ is the set $\mathsf{Min} (\omega)$ of minimal elements of $\{P_i(\omega): i\in \{1, \dots, n\}\}$. More formally,
\begin{equation*}
	\mathsf{Min} (\omega) := \{P_i(\omega): i \in \{1, \dots, n\} \text{ and } P_j(\omega)\subseteq P_i(\omega) \implies P_j(\omega) = P_i(\omega)\text{ for all } j\in \{1, \dots, n\}\}.
\end{equation*}
\begin{lemma}\label{lemma:MINsets}
	Given an $n$-tuple $\mathbf{P} = (P_1, \dots, P_n)$ of binary relations in $\mathcal{P}$, we have that for all $\omega, \omega' \in \Omega$,
	\begin{equation*}
		(\omega, \omega') \in E_{\mathbf{P}} \iff \mathsf{Min} (\omega) = \mathsf{Min} (\omega').
	\end{equation*}
\end{lemma}
\begin{proof}
	It is immediate that $\mathsf{Min} (\omega) = \mathsf{Min} (\omega')$ implies $(\omega, \omega') \in E_{\mathbf{P}}$. To show the other direction, suppose $(\omega, \omega') \in E_{\mathbf{P}}$. Take any $P_i(\omega)\in \mathsf{Min} (\omega)$. Since $(\omega, \omega')\in T_{\mathbf{P}}$, there is a $j$ such that $P_j(\omega')\subseteq P_i(\omega)$. Furthermore, there exists a $k$ such that $P_k(\omega') \in \mathsf{Min} (\omega')$ and $P_k(\omega')\subseteq P_j(\omega')$. Thus $P_k(\omega') \subseteq P_i(\omega)$. Analogously, since $(\omega', \omega)\in T_{\mathbf{P}}$ we can find a $P_{\ell}(\omega) \in \mathsf{Min}(\omega)$ such that $P_{\ell}(\omega) \subseteq P_k(\omega')$. Since $P_i(\omega)$ is a minimal element, we must have $P_{\ell}(\omega) = P_k(\omega') = P_i(\omega)$. Therefore, $\mathsf{Min} (\omega) \subseteq \mathsf{Min} (\omega')$. The reverse set inclusion can be proved analogously.
\end{proof}

\subsection{Proof of Lemma \ref{lemma:Duality}}\label{Proof:Duality}
The argument is split in four independent parts.
\begin{itemize}
	\item $f$ \textit{is injective}. Let $P \neq P'$. This means that there exists a state $\omega$ such that $P(\omega) \neq P'(\omega)$. The latter is equivalent to having $P'(\omega)\setminus P(\omega) \neq \emptyset$ or $P(\omega)\setminus P'(\omega) \neq \emptyset$. Suppose $P'(\omega)\setminus P(\omega) \neq \emptyset$. Call $K:=f(P)$ and $K':=f(P')$. Take the event $A = P(\omega)$. Then we have $\omega \in K(A)$ but $\omega \notin K'(A)$, which means $K \neq K'$. The case $P(\omega)\setminus P'(\omega) \neq \emptyset$ is analogous.
	\item $f$ \textit{is surjective}. Fix $K \in \mathcal{K}^0$. Let $P$ be the possibility relation formed out of $K$ via \eqref{eq:ReprMSZ}. We are going to show that $P$ represents $K$, i.e., $ K = f(P)$. To do this, take any $A\subseteq \Omega$. Suppose $\omega \in K(A)$. It follows immediately from \eqref{eq:ReprMSZ} that $P(\omega) \subseteq A$. Thus, $K \leq f(P)$. Now suppose that $P(\omega) \subseteq A$ for some $\omega \in \Omega$. Since $\omega \in \Omega = K(\Omega)$ by necessitation, it follows from \eqref{eq:ReprMSZ} that $P(\omega)$ is the intersection of a non-empty family of events $B_1, \dots, B_m$, with $m\geq 1$. Besides, we have that $\omega \in \cap_{i=1}^m K(B_i)$. By distributivity, $\cap_{i=1}^m K(B_i) = K(\cap_{i=1}^m B_i)$. Since $\cap_{i=1}^m B_i = P(\omega) \subseteq A$, it follows from monotonicity \eqref{eq:Monot} that $K(\cap_{i=1}^m B_i) \subseteq K(A)$. Therefore, $\omega \in K(A)$ and, as a consequence, $f(P) \leq K$.
	\item $P \subseteq P' \implies f(P') \leq f(P)$. Suppose $P \subseteq P'$, which is equivalent to $P(\omega) \subseteq P'(\omega)$ for all states $\omega$. Call $K:=f(P)$ and $K':=f(P')$. Thus for all events $A$ we have
	\begin{equation*}
		K'(A) = \{ \omega: P'(\omega) \subseteq A \} \subseteq \{ \omega: P(\omega) \subseteq A \} = K(A).
	\end{equation*}
	\item $f(P') \leq f(P) \implies P \subseteq P'$. Call $K:=f(P)$ and $K':=f(P')$. By way of contradiction, suppose $K' \leq K$ but $P \not\subseteq P'$. The latter is equivalent to $P(\omega)\setminus P'(\omega)\neq \emptyset$ for some $\omega$. Choose $B = P'(\omega)$. Clearly, $\omega \in K'(B)$. Since $P(\omega) \not \subseteq P'(\omega)$, we also have $\omega \not\in K(B)$. This contradicts the assumption that $K'(A) \subseteq K(A)$ for all events $A$.
\end{itemize}
\qed

\subsection{Properties of the positive introspection operator}\label{Proof:PosIntrProperties}
The following lemma lists the main properties of the positive introspection operator.

\begin{lemma}\label{lemma:PosIntrProperties}
For all $K \in \mathcal{K}^v$, the following are true.
\begin{itemize} 
\item[(i)] $K^{\bm{+}}$ satisfies veridicality \eqref{eq:K1} and positive introspection \eqref{eq:K2}.
\item[(ii)] $K^{\bm{+}} = K$ if and only if $K$ satisfies positive introspection. 
\item[(iii)] $K\leq K^{\bm{+}}$, i.e., the positive introspection operator $(.)^{\bm{+}}$ is extensive.
\item[(iv)] $K^{\bm{++}} = K^{\bm{+}}$, i.e., the positive introspection operator $(.)^{\bm{+}}$ is idempotent.
\end{itemize}
\end{lemma}

\bigskip

\begin{proof}
\begin{itemize}
\item[(i)] \textit{Veridicality} of $K^{\bm{+}}$ follows from the veridicality of $K$ and the definition of $K^{\bm{+}}$. To show that $K^{\bm{+}}$ satisfies \textit{positive introspection}, notice that if $A\in \mathsf{Img}(K)$, then $K^{\bm{+}}(A) = A$. Consequently, $K^{\bm{+}}(K^{\bm{+}}(A)) = K^{\bm{+}}(A) = A$. If $A \notin \mathsf{Img}(K)$, then $K^{\bm{+}}(A) = K(A)$. Thus $K^{\bm{+}}(K^{\bm{+}}(A)) = K^{\bm{+}}(K(A)) = K(A)$ since $K(A) \in \mathsf{Img}(K)$.
\item[(ii)] If $K^{\bm{+}} = K$, then part (i) above implies that $K$ satisfies positive introspection. To show the other direction, suppose $K$ satisfies positive introspection. Take any event $A\subseteq \Omega$. If $K(A) = A$, it follows immediately from the definition of $K^{\bm{+}}$ that $K^{\bm{+}}(A) = A$. Now suppose $K(A) \neq A$. There are two cases. If $A \notin \mathsf{Img}(K)$, then $K^{\bm{+}}(A) = K(A)$. In the remaining case, $A \in \mathsf{Img}(K)$. This means that there exists an event $B \neq A$ such that $K(B) = A$. By positive introspection, $K(B) = K(K(B))$. Since $K(B) = A$, we get $K(A) = A$, which contradicts the assumption that $K(A) \neq A$.
\item[(iii)] This follows immediately from the definition of $K^{\bm{+}}$ and the axiom of veridicality.
\item[(iv)] By part (i), the revised knowledge operator $K^{\bm{+}}$ satisfies positive introspection. Thus $K^{\bm{++}}= K^{\bm{+}}$ by part (ii).
\end{itemize}
\end{proof}

\subsection{A counterexample to monotonicity}\label{LackMonot}
Neither the positive introspection operator nor the full introspection one is monotone. To see how monotonicity can fail, consider the state space $\Omega = \{\omega_1, \omega_2\}$ and the knowledge operators $K$ and $J$ in the following table. It is easy to check that, although $K \leq J$, we have $K^{\bm{+}} \not\leq J^{\bm{+}}$ and $K^{\bm{\pm}} \not\leq J^{\bm{\pm}}$, thus violating monotonicity.

\begin{table}[h]
\centering
\begin{tabular}{ c | c | c | c | c | c | c}
$A$ &  $K(A)$  & $J(A)$&   $K^{\bm{+}}(A)$ & $J^{\bm{+}}(A)$ &   $K^{\bm{\pm}}(A)$ & $J^{\bm{\pm}}(A)$\\
  \hline			
  $\emptyset$ & $\emptyset$  & $\emptyset$ & $\emptyset$ & $\emptyset$ & $\emptyset$ & $\emptyset$\\
  $\{\omega_1\}$ & $\emptyset$ & $\emptyset$ &  $\{\omega_1\}$ & $\emptyset$ &  $\{\omega_1\}$ & $\emptyset$\\
  $\{\omega_2\}$ & $\emptyset$ & $\emptyset$   & $\emptyset$ & $\emptyset$ & $\{\omega_2\}$ & $\emptyset$\\
  $\{\omega_1,\omega_2\}$ & $\{\omega_1\}$ & $\{\omega_1,\omega_2\}$  & $\{\omega_1\}$ & $\{\omega_1, \omega_2\}$ & $\{\omega_1, \omega_2\}$ & $\{\omega_1, \omega_2\}$
\end{tabular}
\end{table}

\subsection{Properties of the full introspection operator}\label{Proof:FullIntrProperties}
The following lemma lists the main properties of the full introspection operator.

\begin{lemma}\label{lemma:FullIntroProperties}
For all $K \in \mathcal{K}^v$, the following are true.
\begin{itemize}
\item[(i)] $K^{\bm{\pm}}$ satisfies veridicality \eqref{eq:K1}, necessitation \eqref{eq:K0'}, positive introspection \eqref{eq:K2}, and negative introspection \eqref{eq:K3}.

\item[(ii)] $K^{\bm{\pm}} = K$ if and only if $K$ satisfies negative introspection. 

\item[(iii)] $K\leq K^{\bm{\pm}}$, i.e., the full introspection operator $(.)^{\bm{\pm}}$ is extensive.
\item[(iv)] $K^{\bm{\pm \pm}} = K^{\bm{\pm}}$, i.e., the full introspection operator $(.)^{\bm{\pm}}$ is idempotent.
\end{itemize}
\end{lemma}

\bigskip

\begin{proof}
\begin{itemize}
\item[(i)] \textit{Veridicality} of $K^{\bm{\pm}}$ follows from the veridicality of $K$ and the definition of $K^{\bm{\pm}}$. To show that $K^{\bm{\pm}}$ satisfies \textit{negative introspection}, suppose $A \in \mathsf{Img}(K)\cup \mathsf{Img}(\lnot K)$. Then $K^{\bm{\pm}}(A) = A$ and $\lnot K^{\bm{\pm}}(A) = \lnot A$. Moreover, $K^{\bm{\pm}} (\lnot K^{\bm{\pm}}(A)) = K^{\bm{\pm}}(\lnot A)$. If $A \in \mathsf{Img}(K)$, then $\lnot A \in \mathsf{Img}(\lnot K)$ and $K^{\bm{\pm}}(\lnot A) = \lnot A$. Analogously, if $A \in \mathsf{Img}(\lnot K)$, then $\lnot A \in \mathsf{Img}(K)$ and $K^{\bm{\pm}}(\lnot A) = \lnot A$. Now suppose $A \notin \mathsf{Img}(K)\cup \mathsf{Img}(\lnot K)$. By definition, $K^{\bm{\pm}}(A) = K(A)$ and, consequently, $\lnot K^{\bm{\pm}}(A) = \lnot K(A)$. Furthermore, $K^{\bm{\pm}} (\lnot K^{\bm{\pm}}(A)) = K^{\bm{\pm}} (\lnot K(A)) = \lnot K(A)$, since $\lnot K(A) \in \mathsf{Img}(\lnot K)$. As for the remaining properties, it is well known that veridicality and negative introspection jointly imply  \textit{necessitation} and \textit{positive introspection} \citep[pp. 269-270]{aumann1999interactive1}.
\item[(ii)] If $K^{\bm{\pm}} = K$, then part (i) above implies that $K$ satisfies negative introspection. To show the other direction, suppose $K$ satisfies negative introspection. Take any $A\subseteq \Omega$. If $K(A) = A$, then $K^{\bm{\pm}}(A) = A$ by the definition of $K^{\bm{\pm}}$. Now suppose $K(A) \neq A$. If $A \notin \mathsf{Img}(K) \cup \mathsf{Img}(\lnot K)$, then $K^{\bm{\pm}}(A) = K(A)$. If $A \in \mathsf{Img}(K) \cup \mathsf{Img}(\lnot K)$, there are two subcases. First, $A \in \mathsf{Img}(\lnot K)$. This means that there exists an event $B$ such that $\lnot K(B) = A$. By negative introspection, $\lnot K(B) = K(\lnot K(B))$. Since $A = \lnot K(B)$, we get $K(A) = A$, which contradicts the assumption $K(A) \neq A$. In the remaining subcase, $A \in \mathsf{Img}(K)$. Hence there is an event $B$ such that $A = K(B)$, which is equivalent to $\lnot A = \lnot K(B)$. By negative introspection, $\lnot K(B) = K(\lnot K(B))$ and, consequently, $\lnot A = K(\lnot A)$, which is equivalent to $A = \lnot K(\lnot A)$. By negative introspection again, $\lnot K(\lnot A) = K(\lnot K(\lnot A))$ and, consequently, $K(A) = A$, so reaching a contradiction.
\item[(iii)] This follows from the definition of $K^{\bm{\pm}}$ and the fact that $K$ satisfies veridicality.
\item[(iv)] By part (i), the revised knowledge operator $K^{\bm{\pm}}$ satisfies negative introspection. Thus $K^{\bm{\pm\pm}} = K^{\bm{\pm}}$ by part (ii).
\end{itemize}
\end{proof}

\subsection{Properties of the distributive closure}\label{Proof:DistrClosProperties}
The following lemma lists the main properties of the distributive closure.

\begin{lemma}\label{lemma:DistrClosProperties}
For all $K, J \in \mathcal{K}^v$, the following are true.
\begin{itemize}
\item[(i)] $K^{\mathbf{d}}$ satisfies veridicality \eqref{eq:K1} and distributivity \eqref{eq:K0}.

\item[(ii)] $K^{\mathbf{d}} = K$ if and only if $K$ satisfies distributivity. 

\item[(iii)] $K\leq K^{\mathbf{d}}$, i.e., the distributive closure $(.)^{\mathbf{d}}$ is extensive.
\item[(iv)] $K^{\mathbf{dd}} = K^{\mathbf{d}}$, i.e., the distributive closure $(.)^{\mathbf{d}}$ is idempotent.
\item[(v)] $K \leq J$ implies $K^{\mathbf{d}} \leq J^{\mathbf{d}}$, i.e., the distributive closure $(.)^{\mathbf{d}}$ is monotone.
\item[(vi)] If $K$ satisfies positive introspection, so does $K^{\mathbf{d}}$.
\item[(vii)] If $K$ satisfies negative introspection, so does $K^{\mathbf{d}}$.
\end{itemize} 
\end{lemma}

\bigskip

\begin{proof}
\begin{itemize}
\item[(i)] \textit{Veridicality} of $K^{\mathbf{d}}$ follows from the veridicality of $K$ and the definition of $K^{\mathbf{d}}$. To show that $K^{\mathbf{d}}$ satisfies \textit{distributivity}, notice that $K^{\mathbf{d}}$ satisfies monotonicity \eqref{eq:Monot}. Hence, $K^{\mathbf{d}}(A \cap B) \subseteq K^{\mathbf{d}}(A) \cap K^{\mathbf{d}}(B)$ for all $A, B\subseteq \Omega$. To show the reverse inclusion, take any $A, B\subseteq \Omega$. The case where $K^{\mathbf{d}}(A) \cap K^{\mathbf{d}}(B) = \emptyset$ is trivial. So suppose $\omega \in K^{\mathbf{d}}(A) \cap K^{\mathbf{d}}(B) \neq \emptyset$. By the definition of $K^{\mathbf{d}}$, there exist events $C_1, \dots, C_n, D_1, \dots, D_m \subseteq \Omega$ such that $\cap_{i=1}^n C_i \subseteq A$, $\cap_{j=1}^m D_j \subseteq B$, and
\begin{align}
\omega &\in \left(\cap_{i=1}^n K(C_i)\right) \cap \left(\cap_{j=1}^m K(D_j)\right)\nonumber\\
&= K(C_1) \cap \cdots \cap K(C_n) \cap K(D_1) \cap \cdots \cap K(D_m). \label{eq:DistrProof}
\end{align}
Since $\cap_{i=1}^n C_i \subseteq A$ and $\cap_{j=1}^m D_j \subseteq B$, we have
\begin{equation*}
C_1 \cap \cdots \cap C_n \cap D_1 \cap \cdots \cap D_m \subseteq A \cap B.
\end{equation*}
Hence it follows from the definition of $K^{\mathbf{d}}$ that the set \eqref{eq:DistrProof} is included in $K^{\mathbf{d}}(A\cap B)$, so proving that $K^{\mathbf{d}}$ satisfies distributivity.
\item[(ii)] If $K^{\mathbf{d}} = K$, then part (i) above implies that $K$ satisfies distributivity. To show the other direction, suppose $K$ satisfies distributivity. Fix $A\subseteq \Omega$ and take any sequence of events $B_1, \dots, B_n$ such that $\cap_{i=1}^n B_i \subseteq A$. Since $K$ satisfies distributivity, and consequently monotonicity, we have
\begin{equation*}
\cap_{i=1}^n K(B_i) = K(\cap_{i=1}^n B_i) \subseteq K(A).
\end{equation*}
This means that $K^{\mathbf{d}} (A)$ is the union of subsets of $K(A)$. Therefore, $K^{\mathbf{d}}(A) \subseteq K(A)$. In addition, it follows easily from the definition of $K^{\mathbf{d}}$ that $K(A) \subseteq K^{\mathbf{d}}(A) $. Thus we can conclude that $K(A) = K^{\mathbf{d}}(A)$ for all $A\subseteq \Omega$.
\item[(iii)] This follows immediately from the definition of $K^{\mathbf{d}}$.
\item[(iv)] The revised knowledge operator $K^{\mathbf{d}}$ satisfies distributivity by part (i). Thus $K^{\mathbf{dd}} = K^{\mathbf{d}}$ by part (ii).
\item[(v)] Suppose $K \leq J$, i.e., $K(A) \subseteq J(A)$ for all $A\subseteq \Omega$. Then $\cap_{i=1}^n K(B_i) \subseteq \cap_{i=1}^n J(B_i)$ for all sequences $B_1, \dots, B_n \subseteq \Omega$. Therefore, $K^{\mathbf{d}} \leq J^{\mathbf{d}}$.
\item[(vi)] Suppose $K \in \mathcal{K}^v$ satisfies positive introspection. Fix $A\subseteq \Omega$. If $K^{\mathbf{d}}(A)$ is the empty set, it is obvious that $K^{\mathbf{d}}(A) \subseteq K^{\mathbf{d}}(K^{\mathbf{d}}(A))$. Now suppose $K^{\mathbf{d}}(A) \neq \emptyset$ and take $\omega \in K^{\mathbf{d}}(A)$. By the definition of $K^{\mathbf{d}}$, there exist events $B_1, \dots, B_n \subseteq \Omega$ such that $\cap_{i=1}^n B_i \subseteq A$ and $\omega \in \cap_{i=1}^n K(B_i) \subseteq K^{\mathbf{d}}(A)$. By positive introspection and veridicality, we have $\cap_{i=1}^n K(B_i) = \cap_{i=1}^n K(K(B_i))$. By using the fact that $\cap_{i=1}^n K(B_i) \subseteq K^{\mathbf{d}}(A)$ and the definition of $K^{\mathbf{d}}$, we can conclude that
\begin{equation*}
	\omega \in  \cap_{i=1}^n K(K(B_i)) \subseteq K^{\mathbf{d}}(K^{\mathbf{d}}(A)).
\end{equation*}
\item[(vii)] Suppose $K \in \mathcal{K}^v$ satisfies negative introspection. Fix $A\subseteq \Omega$. If $\lnot K^{\mathbf{d}}(A)$ is the empty set, then it is obvious that $\lnot K^{\mathbf{d}}(A) \subseteq K^{\mathbf{d}}(\lnot K^{\mathbf{d}}(A))$. Now suppose $\lnot K^{\mathbf{d}}(A)\neq \emptyset$ and take $\omega \in \lnot K^{\mathbf{d}}(A)$. By the definition of $K^{\mathbf{d}}$ and De Morgan's laws, the event $\lnot K^{\mathbf{d}}(A)$ is
\begin{equation}\label{eq:NotKdA}
\lnot K^{\mathbf{d}}(A) = \bigcap \left\lbrace \cup_{i=1}^n \lnot K(B_i): B_1, \dots, B_n \subseteq \Omega, \; \cap_{i=1}^n B_i \subseteq A, \; n \geq 1\right\rbrace.
\end{equation}

As per \eqref{eq:NotKdA}, the event $\lnot K^{\mathbf{d}}(A)$ is the intersection of infinitely many events of the form $\cup_{i=1}^n \lnot K(B_i)$. Since the state space $\Omega$ is finite, there are only finitely many such events that are not equal to one another. Thus $\lnot K^{\mathbf{d}}(A)$ is fully determined by the intersection of finitely many events of the form $\cup_{i=1}^n \lnot K(B_i)$. In addition, recall that set intersection distributes over union and that for any two finite families $\{A_j\}_{j\in J}$ and $\{A_k\}_{k\in K}$ we have
\begin{equation}\label{eq:DistrIntersection}
\left(\cup_{j\in J} A_j\right) \cap \left(\cup_{k\in K} A_k\right) = \bigcup_{(j,k)\in J\times K} (A_j \cap A_k).
\end{equation}

Now, by repeated application of \eqref{eq:DistrIntersection}, the event $\lnot K^{\mathbf{d}}(A)$ can be written as the union of finitely many events of the form $\cap_{i=1}^m \lnot K(C_i)$. Since $\omega \in \lnot K^{\mathbf{d}}(A)$ by assumption, there must exist a sequence $C_1, \dots, C_m \subseteq \Omega$ for which
\begin{equation*}
\omega \in \cap_{i=1}^m \lnot K(C_i) \subseteq \lnot K^{\mathbf{d}}(A).
\end{equation*}

By negative introspection and veridicality, we have $\cap_{i=1}^m \lnot K(C_i) = \cap_{i=1}^m K(\lnot K(C_i))$. By using the fact that $\cap_{i=1}^m \lnot K(C_i) \subseteq \lnot K^{\mathbf{d}}(A)$ and by the definition of $K^{\mathbf{d}}$ we can conclude that
\begin{equation*}
\omega \in \cap_{i=1}^m \lnot K(C_i) = \cap_{i=1}^m K(\lnot K(C_i)) \subseteq K^{\mathbf{d}} (\lnot K^{\mathbf{d}}(A)).
\end{equation*}
\end{itemize}
\end{proof}

\subsection{Proof of Proposition \ref{Prop:FeasRev}}\label{Proof:FeasRev}
The proof relies on the following two lemmata.

\begin{lemma}\label{lemma:K2}
Let $K_1, \dots, K_n$ be $n\geq 1$ knowledge operators in  $\mathcal{K}^2$ and let $\mathbf{P} = \left(P_{1}, \dots, P_{n}\right)$ be the profile of possibility relations in $\mathcal{P}^2$ that represent those $n$ knowledge operators. Then we have:
\begin{equation*}
(\cup_{i} K_i)^{\mathbf{d}} = (\cup_{i} K_i)^{\mathbf{d}\bm{+}\mathbf{d}} = (\cup_{i} K_i)^{\bm{+}\mathbf{d}} = (\cup_{i} K_i^{\bm{+}\mathbf{d}})^{\mathbf{d}} \leq (\cup_{i} K_i)^{\mathbf{d}\bm{\pm}\mathbf{d}} = (\cup_{i} K_i)^{\bm{\pm}\mathbf{d}} = (\cup_{i} K_i^{\bm{\pm}\mathbf{d}})^{\mathbf{d}}
\end{equation*}
and, dually,
\begin{equation}\label{eq:K2claim}
\cap_{i} P_{i} = T_{\cap_{i} P_{i}} = T_{\mathbf{P}} = \cap_i T_{P_i} \supseteq E_{\cap_{i} P_{i}} = E_{\mathbf{P}} = \cap_i E_{P_i}.
\end{equation}
\end{lemma}
\begin{proof}[Proof of Lemma \ref{lemma:K2}] By Lemma \ref{lemma:Duality} and Proposition \ref{Prop:Representation}, it is enough to prove the dual statement \eqref{eq:K2claim}. Suppose $(\omega, \omega')\in \cap_{i} P_{i}$. Then $(\omega, \omega') \in P_i$ for all $i\in \{1, \dots, n\}$. Since each $P_i$ is transitive by assumption, $(\omega, \omega') \in P_i$ implies $P_i(\omega') \subseteq P_i(\omega)$. Thus $(\omega, \omega') \in \cap_i T_{P_i}$ and $\cap_{i} P_{i} \subseteq \cap_i T_{P_i}$. By Lemma \ref{lemma:K1}, we can conclude $\cap_{i} P_{i} = T_{\cap_{i} P_{i}} = T_{\mathbf{P}} = \cap_i T_{P_i} \supseteq E_{\cap_{i} P_{i}}$. Since $T_{\cap_{i} P_{i}} = T_{\mathbf{P}} = \cap_i T_{P_i}$, we also have $E_{\cap_{i} P_{i}} = E_{\mathbf{P}} = \cap_i E_{P_i}$, so ending the proof.
\end{proof}

\begin{lemma}\label{lemma:K3}
Let $K_1, \dots, K_n$ be $n\geq 1$ knowledge operators in  $\mathcal{K}^3$, and let $\mathbf{P} = \left(P_{1}, \dots, P_{n}\right)$ be the profile of possibility relations in $\mathcal{P}^3$ that represent those $n$ knowledge operators. Then we have:
\begin{equation*}
(\cup_{i} K_i)^{\mathbf{d}} = (\cup_{i} K_i)^{\mathbf{d}\bm{+}\mathbf{d}} = (\cup_{i} K_i)^{\bm{+}\mathbf{d}} = (\cup_{i} K_i^{\bm{+}\mathbf{d}})^{\mathbf{d}} = (\cup_{i} K_i)^{\mathbf{d}\bm{\pm}\mathbf{d}} = (\cup_{i} K_i)^{\bm{\pm}\mathbf{d}} = (\cup_{i} K_i^{\bm{\pm}\mathbf{d}})^{\mathbf{d}}
\end{equation*}
and, dually,
\begin{equation}\label{eq:K3claim}
\cap_{i} P_{i} = T_{\cap_{i} P_{i}} = T_{\mathbf{P}} = \cap_i T_{P_i} = E_{\cap_{i} P_{i}} = E_{\mathbf{P}} = \cap_i E_{P_i}.
\end{equation}
\end{lemma}
\begin{proof}[Proof of Lemma \ref{lemma:K3}] As in the previous lemma, it is enough to show \eqref{eq:K3claim}. Since each possibility relation is an equivalence relation by assumption, we have $\cap_i P_i \subseteq \cap_i E_{P_i}$. Combining the latter with Lemma \ref{lemma:K2} yields \eqref{eq:K3claim}.
\end{proof}

Now Proposition \ref{Prop:FeasRev} can be proved as follows.

\begin{itemize}
\item[(i)] Suppose agent $i$ is non-introspective. Since the distributive closure is idempotent, we have $\mathbf{K}^{\theta_i} = (\cup_{j} K_j)^{\mathbf{d}}$ when $i$ forms $\mathbf{K}^{\theta_i}$ according to part (iii-a) of Assumption \ref{Assm:Rev}. If $i$ revises $\mathbf{K}$ according to part (iii-b), then $\mathbf{K}^{\theta_i} = (\cup_{j} K_j^{\mathbf{d}})^{\mathbf{d}}$. Since each $K_j$ in $\mathbf{K}$ satisfies distributivity by assumption, we have $(\cup_{j} K_j^{\mathbf{d}})^{\mathbf{d}} = (\cup_{j} K_j)^{\mathbf{d}}$ by Lemma \ref{lemma:DistrClosProperties}.
\item[(ii)] Suppose $i$ is positively introspective. When she revises $\mathbf{K}$ according to part (iii-a) of Assumption \ref{Assm:Rev}, two cases are possible. In the first, knowledge revision starts by applying the positive introspection operator, so obtaining $(\cup_j K_j)^{\bm{+}}$. By Lemma \ref{lemma:PosIntrProperties}, the knowledge operator $(\cup_j K_j)^{\bm{+}}$ satisfies positive introspection; in addition, $(\cup_j K_j)^{\bm{+}}$ cannot be further expanded through the positive introspection operator because the latter is idempotent. However, $(\cup_j K_j)^{\bm{+}}$ may or may not satisfy distributivity. If it does, then $(\cup_j K_j)^{\bm{+}} = (\cup_j K_j)^{\bm{+}\mathbf{d}}$ by Lemma \ref{lemma:DistrClosProperties}; if it does not, then $(\cup_j K_j)^{\bm{+}\mathbf{d}}$ is a strict expansion of $(\cup_j K_j)^{\bm{+}}$. The operator $(\cup_j K_j)^{\bm{+}\mathbf{d}}$ always satisfies distributivity and, by Lemma \ref{lemma:DistrClosProperties}, positive introspection too. Thus one cannot strictly expand $(\cup_j K_j)^{\bm{+}\mathbf{d}}$ through further applications of the distributive closure or of the positive introspection operator. Therefore, we can conclude that $\mathbf{K}^{\theta_i} = (\cup_{j} K_j)^{\bm{+}\mathbf{d}}$. In the second case, knowledge revision starts by applying the distributive closure to $\cup_j K_j$, so obtaining $(\cup_j K_j)^{\mathbf{d}}$. The latter may or may not satisfy positive introspection. If it does, then $(\cup_j K_j)^{\mathbf{d}} = (\cup_{j} K_j)^{\mathbf{d}\bm{+}\mathbf{d}}$ by Lemmata \ref{lemma:PosIntrProperties} and \ref{lemma:DistrClosProperties}. If it does not, then $(\cup_j K_j)^{\mathbf{d}\bm{+}}$ is a strict expansion of $(\cup_j K_j)^{\mathbf{d}}$. Now, the knowledge operator $(\cup_j K_j)^{\mathbf{d}\bm{+}}$ does not necessarily satisfy distributivity. If it does, then $(\cup_j K_j)^{\mathbf{d}\bm{+}}= (\cup_{j} K_j)^{\mathbf{d}\bm{+}\mathbf{d}}$ by Lemma \ref{lemma:DistrClosProperties}. Otherwise, $(\cup_j K_j)^{\mathbf{d}\bm{+}}$ is strictly expanded by $(\cup_j K_j)^{\mathbf{d}\bm{+}\mathbf{d}}$. The operator $(\cup_j K_j)^{\mathbf{d}\bm{+}\mathbf{d}}$ always satisfies distributivity and, by Lemma \ref{lemma:DistrClosProperties}, positive introspection too. Thus one cannot strictly expand $(\cup_j K_j)^{\mathbf{d}\bm{+}\mathbf{d}}$ through further applications of the distributive closure or of the positive introspection operator. Therefore, we can conclude that in this case $\mathbf{K}^{\theta_i} = (\cup_{j} K_j)^{\mathbf{d}\bm{+}\mathbf{d}}$.

It remains to show what is $\mathbf{K}^{\theta_i}$ when $i$ revises $\mathbf{K}$ according to part (iii-b) of Assumption \ref{Assm:Rev}. In this case, $i$ applies the distributive closure and the positive introspection operator to each $K_j$. Since each $K_j$ satisfies distributivity by assumption, and since the distributive closure preserves positive introspection by Lemma \ref{lemma:DistrClosProperties}, we have $K_j^{\bm{+}\mathbf{d}} = K_j^{\mathbf{d}\bm{+}\mathbf{d}}$ and no strict expansion of the latter can be made through further applications of $(.)^{\mathbf{d}}$ or $(.)^{\bm{+}}$. Thus the pointwise union of all $n$ revised operators from $\mathbf{K}$ is equal to $\cup_j K_j^{\bm{+}\mathbf{d}}$. The latter operator does not necessarily satisfy distributivity. If it does, then $\cup_j K_j^{\bm{+}\mathbf{d}} = (\cup_{j} K_i^{\bm{+}\mathbf{d}})^{\mathbf{d}}$ by Lemma \ref{lemma:DistrClosProperties}. Otherwise, $\cup_j K_j^{\bm{+}\mathbf{d}}$ is strictly expanded by $(\cup_j K_j^{\bm{+}\mathbf{d}})^{\mathbf{d}}$. Since each $
K_j^{\bm{+}\mathbf{d}}$ belongs to $\mathcal{K}^2$, it follows from Lemma \ref{lemma:K2} that no strict expansion of $(\cup_{j} K_j^{\bm{+}\mathbf{d}})^{\mathbf{d}}$ can be made through the distributive closure or the positive introspection operator. Therefore, we can conclude that $\mathbf{K}^{\theta_i} = (\cup_{j} K_j^{\bm{+}\mathbf{d}})^{\mathbf{d}}$.
\item[(iii)] One can prove the case of a fully introspective agent by following, \textit{mutatis mutandis}, the argument of part (ii) above. Note that Lemma \ref{lemma:K3} is the relevant lemma for the proof instead of Lemma \ref{lemma:K2}.
\end{itemize}
\qed

\subsection{Proof of Proposition \ref{Prop:Representation}}\label{Proof:Representation}
The proof relies on the following two lemmata.
\begin{lemma}\label{Lemma:Equiv+}
Let $K_1, \dots, K_n$ be $n\geq 1$ knowledge operators in $\mathcal{K}^1$, and let $\mathbf{P}=\left(P_{1}, \dots, P_{n}\right)$ be the profile of relations in $\mathcal{P}^1$ that represent those $n$ knowledge operators. In addition, let $Q$ be the relation that represents $(\cup_{i=1}^n K_i)^{\bm{+}\mathbf{d}}$. For all $\omega, \omega' \in \Omega$, the following are equivalent:
\begin{itemize}
\item[(i)] $(\omega, \omega') \in T_{\mathbf{P}}$;
\item[(ii)] for all $A \subseteq \Omega$, $\omega \in \cup_{i=1}^n K_i(A) \implies \omega' \in \cup_{i=1}^n K_i(A)$;
\item[(iii)] for all $A \subseteq \Omega$, $\omega \in (\cup_{i=1}^n K_i)^{\bm{+}}(A) \implies \omega' \in (\cup_{i=1}^n K_i)^{\bm{+}}(A)$;
\item[(iv)] for all $A \subseteq \Omega$, $\omega \in (\cup_{i=1}^n K_i)^{\bm{+}\mathbf{d}}(A) \implies \omega' \in (\cup_{i=1}^n K_i)^{\bm{+}\mathbf{d}}(A)$;
\item[(v)] $(\omega, \omega') \in Q$.
\end{itemize}
\end{lemma}
\begin{proof}[Proof of Lemma \ref{Lemma:Equiv+}]
\begin{itemize}
\item $\bm{(i) \implies (ii).}$ Suppose $(\omega, \omega') \in T_{\mathbf{P}}$. If $\omega \in \cup_{i=1}^n K_i(A)$ for some $A \subseteq \Omega$, then $\omega \in K_h(A)$ for some $h\in \{1, \dots, n\}$. This means that $P_{h}(\omega) \subseteq A$. By the definition of the left $n$-ary trace, there exists a $j \in \{1, \dots, n\}$ such that $P_{j}(\omega') \subseteq P_{h}(\omega)$. Therefore, $\omega' \in K_j(A) \subseteq \cup_{i=1}^n K_i(A)$.
\item $\bm{(ii) \implies (i).}$ Suppose $(ii)$ holds. It is clear that, for all $i\in \{1, \dots, n\}$, $\omega \in K_i(P_{i}(\omega)) \subseteq \cup_{j=1}^n K_j(P_{i}(\omega))$. It follows from $(ii)$ that $\omega' \in \cup_{j=1}^n K_j(P_{i}(\omega))$ too. The latter implies that $\omega' \in K_h(P_{i}(\omega))$ for some $h \in \{1, \dots, n\}$, which is equivalent to $P_{h}(\omega') \subseteq P_{i}(\omega)$. Thus $(\omega, \omega') \in T_{\mathbf{P}}$.
\item $\bm{(ii) \implies (iii).}$ This follows from the definition of $(\cup_{i=1}^n K_i)^{\bm{+}}$.
\item $\bm{(iii) \implies (ii).}$ This follows from the fact that for all $A\subseteq \Omega$, the event $\cup_{i=1}^n K_i(A)$ is a fixed point of $(\cup_{i=1}^n K_i)^{\bm{+}}$.
\item $\bm{(iii) \implies (iv).}$ This follows from the definition of $(\cup_{i=1}^n K_i)^{\bm{+}\mathbf{d}}$.
\item $\bm{(iv) \implies (iii).}$ Suppose $(iv)$ holds. For all $A\subseteq \Omega$, the event $(\cup_{i=1}^n K_i)^{\bm{+}}(A)$ is a fixed point of $(\cup_{i=1}^n K_i)^{\bm{+}}$. In addition, since $(.)^{\mathbf{d}}$ is extensive, and since $(\cup_{i=1}^n K_i)^{\bm{+}\mathbf{d}}$ satisfies veridicality, we can conclude that
\begin{equation*}
\left(\cup_{i=1}^n K_i\right)^{\bm{+}}(A) = \left(\cup_{i=1}^n K_i\right)^{\bm{+}\mathbf{d}}\left((\cup_{i=1}^n K_i)^{\bm{+}}(A)\right),
\end{equation*}
from which $(iii)$ follows.
\item $\bm{(iv) \implies (v).}$ Suppose $(iv)$ holds. Since $Q$ represents $\left(\cup_{i=1}^n K_i\right)^{\bm{+}\mathbf{d}}$ by assumption,
\begin{equation*}
Q(\omega) = \bigcap\{A: \omega \in (\cup_{i=1}^n K_i)^{\bm{+}\mathbf{d}}(A)\} \supseteq \bigcap\{A: \omega' \in (\cup_{i=1}^n K_i)^{\bm{+}\mathbf{d}}(A)\} = Q(\omega').
\end{equation*}
Since $Q\in \mathcal{P}^2$, $Q$ is reflexive. Thus we can conclude that $\omega' \in Q(\omega)$, i.e., $(\omega, \omega')\in Q$.
\item $\bm{(v) \implies (iv).}$ Suppose $(\omega, \omega')\in Q$. Since $Q$ is transitive, $\omega' \in Q(\omega)$ implies $Q(\omega')\subseteq Q(\omega)$, which in turn implies $(iv)$.
\end{itemize}
\end{proof}

\begin{lemma}\label{Lemma:Equiv+and-}
Let $K_1, \dots, K_n$ be $n\geq 1$ knowledge operators in $\mathcal{K}^1$, and let $\mathbf{P}=\left(P_{1}, \dots, P_{n}\right)$ be the profile of relations in $\mathcal{P}^1$ that represent those $n$ knowledge operators. In addition, let $Q$ be the relation that represents $(\cup_{i=1}^n K_i)^{\bm{\pm}\mathbf{d}}$. For all $\omega, \omega' \in \Omega$, the following are equivalent:
\begin{itemize}
\item[(i)] $(\omega, \omega') \in E_{\mathbf{P}}$;
\item[(ii)] for all $A \subseteq \Omega$, $\omega \in \cup_{i=1}^n K_i(A) \iff \omega' \in \cup_{i=1}^n K_i(A)$;
\item[(iii)] for all $A \subseteq \Omega$, $\omega \in (\cup_{i=1}^n K_i)^{\bm{\pm}}(A) \iff \omega' \in (\cup_{i=1}^n K_i)^{\bm{\pm}}(A)$;
\item[(iv)] for all $A \subseteq \Omega$, $\omega \in (\cup_{i=1}^n K_i)^{\bm{\pm}\mathbf{d}}(A) \iff \omega' \in (\cup_{i=1}^n K_i)^{\bm{\pm}\mathbf{d}}(A)$;
\item[(v)] $(\omega, \omega') \in Q$.
\end{itemize}
\end{lemma}
\begin{proof}[Proof of Lemma \ref{Lemma:Equiv+and-}]
\begin{itemize}
\item $\bm{(i) \iff (ii).}$ This follows immediately from the equivalence between (i) and (ii) in Lemma \ref{Lemma:Equiv+}.
\item $\bm{(ii) \implies (iii).}$ This follows from the definition of $(\cup_{i=1}^n K_i)^{\bm{\pm}}$.
\item $\bm{(iii) \implies (ii).}$ This follows from the fact that, for all $A\subseteq \Omega$, the event $\cup_{i=1}^n K_i(A)$ is a fixed point of $(\cup_{i=1}^n K_i)^{\bm{\pm}}$.
\item $\bm{(iii) \implies (iv).}$ This follows from the definition of $(\cup_{i=1}^n K_i)^{\bm{\pm}\mathbf{d}}$.
\item $\bm{(iv) \implies (iii).}$ Suppose $(iv)$ holds. For all $A\subseteq \Omega$, the event $(\cup_{i=1}^n K_i)^{\bm{\pm}}(A)$ is a fixed point of $(\cup_{i=1}^n K_i)^{\bm{\pm}}$. In addition, since $(.)^{\mathbf{d}}$ is extensive, and since $(\cup_{i=1}^n K_i)^{\bm{\pm}\mathbf{d}}$ satisfies veridicality, we can conclude that
\begin{equation*}
\left(\cup_{i=1}^n K_i\right)^{\bm{\pm}}(A) = \left(\cup_{i=1}^n K_i\right)^{\bm{\pm}\mathbf{d}}\left((\cup_{i=1}^n K_i)^{\bm{\pm}}(A)\right),
\end{equation*}
from which $(iii)$ follows.
\item $\bm{(iv) \implies (v).}$ Suppose $(iv)$ holds. Since $Q$ represents $\left(\cup_{i=1}^n K_i\right)^{\bm{\pm}\mathbf{d}}$ by assumption,
\begin{equation*}
Q(\omega) = \bigcap\{A: \omega \in (\cup_{i=1}^n K_i)^{\bm{\pm}\mathbf{d}}(A)\} = \bigcap\{A: \omega' \in (\cup_{i=1}^n K_i)^{\bm{\pm}\mathbf{d}}(A)\} = Q(\omega').
\end{equation*}
Since $Q\in \mathcal{P}^3$, $Q$ is reflexive. Thus we can conclude that $\omega' \in Q(\omega)$, i.e., $(\omega, \omega')\in Q$.
\item $\bm{(v) \implies (iv).}$ Suppose $(\omega, \omega')\in Q$. Since $Q$ is an equivalence relation, $\omega' \in Q(\omega)$ implies $Q(\omega') = Q(\omega)$, which in turn implies $(iv)$.
\end{itemize}
\end{proof}

Now Proposition \ref{Prop:Representation} can be proved as follows.

\begin{itemize}
\item[(i)] It is clear that $(\cup_{i=1}^n K_i)^{\mathbf{d}} \in \mathcal{K}^1$ and $\cap_{i=1}^n P_{i} \in \mathcal{P}^1$. Call $L$ the knowledge operator represented by $\cap_{i=1}^n P_{i}$. By \eqref{eq:ReprMSZ}, for all $A\subseteq \Omega$,
\begin{equation*}
L(A) = \left\lbrace \omega: \cap_{i=1}^n P_{i}(\omega) \subseteq A\right\rbrace.
\end{equation*}
We need to show that $(\cup_{i=1}^n K_i)^{\mathbf{d}} = L$. Take any $A \subseteq \Omega$. First we show that $(\cup_{i=1}^n K_i)^{\mathbf{d}}(A) \subseteq L(A)$. Suppose $\omega \in (\cup_{i=1}^n K_i)^{\mathbf{d}}(A)$. This means that there exist events $B_1, \dots, B_m$ such that $\cap_{j=1}^m B_j \subseteq A$ and
\begin{equation}\label{eq:RepProof1}
\omega \in \bigcap_{j = 1}^m \left(\cup_{i=1}^n K_i(B_j)\right).
\end{equation}
Since each $P_{i}$ represents $K_i$, \eqref{eq:RepProof1} is equivalent to
\begin{equation*}
\omega \in \bigcap_{j = 1}^m \left(\cup_{i=1}^n \left\lbrace \omega' : P_{i} (\omega') \subseteq B_j\right\rbrace\right),
\end{equation*}
from which we see that, for all $j \in \{1, \dots, m\}$, there exists a $i \in \{1, \dots, n\}$ such that $P_{i}(\omega) \subseteq B_j$. Therefore, 
\begin{equation*}
\cap_{i=1}^n P_{i}(\omega) \subseteq \cap_{j=1}^m B_j \subseteq A,
\end{equation*}
so proving that $\omega \in L(A)$.\\
To show the reverse set inclusion, suppose $\omega \in L(A)$, so that $\cap_{i=1}^n P_{i}(\omega) \subseteq A$. It is clear that $\omega \in K_i(P_{i}(\omega))$ for all $i \in \{1, \dots, n\}$. Therefore, we have
\begin{equation*}
\omega \in \cap_{i = 1}^n K_i(P_{i}(\omega)) \subseteq \bigcap_{i=1}^n \left(\cup_{j=1}^n K_j(P_{i}(\omega))\right) \subseteq (\cup_{i=1}^n K_i)^{\mathbf{d}} (A).
\end{equation*}
\item[(ii)] The knowledge operator $(\cup_{i=1}^n K_i)^{\bm{+}\mathbf{d}}$ belongs to $\mathcal{K}^2$ and is represented by the binary relation $Q \in \mathcal{P}^2$ defined by \eqref{eq:ReprMSZ}. By Lemma \ref{Lemma:Equiv+}, $(\omega, \omega')\in Q$ if and only if $(\omega, \omega')\in T_{\mathbf{P}}$ for all $\omega, \omega' \in \Omega$. Thus $Q = T_{\mathbf{P}}$.
\item[(iii)] It is clear that $(\cup_{i=1}^n K_i)^{\mathbf{d}} \in \mathcal{K}^1$. By part (i) of Proposition \ref{Prop:Representation}, $(\cup_{i=1}^n K_i)^{\mathbf{d}}$ is represented by $\cap_{i=1}^n P_{i}$. Then the result follows immediately from part (ii) of the same proposition.
\item[(iv)] The knowledge operators $K_1^{\bm{+}\mathbf{d}}, \dots, K_n^{\bm{+}\mathbf{d}}$ belong to $\mathcal{K}^2$. By part (ii) of Proposition \ref{Prop:Representation}, each $K_i^{\bm{+}\mathbf{d}}$ is represented by $T_{P_i}$, i.e., the left trace of $P_i$. By part (i) of the same proposition, the operator $(\cup_i K_i^{\bm{+}\mathbf{d}})^{\mathbf{d}}$ is represented by $\cap_i T_{P_i}$.
\item[(v)] The knowledge operator $(\cup_{i=1}^n K_i)^{\bm{\pm}\mathbf{d}}$ belongs to $\mathcal{K}^3$ and is represented by the binary relation $Q \in \mathcal{P}^3$ defined by \eqref{eq:ReprMSZ}. By Lemma \ref{Lemma:Equiv+and-}, $(\omega, \omega')\in Q$ if and only if $(\omega, \omega')\in E_{\mathbf{P}}$ for all $\omega, \omega' \in \Omega$. Thus $Q = E_{\mathbf{P}}$.
\item[(vi)] It is clear that $(\cup_{i=1}^n K_i)^{\mathbf{d}} \in \mathcal{K}^1$. By part (i), $(\cup_{i=1}^n K_i)^{\mathbf{d}}$ is represented by $\cap_{i=1}^n P_{i}$. Then the result follows immediately from part (iv) of Proposition \ref{Prop:Representation}.
\item[(vii)] The knowledge operators $K_1^{\bm{\pm}\mathbf{d}}, \dots, K_n^{\bm{\pm}\mathbf{d}}$ belong to $\mathcal{K}^3$. By part (v) of Proposition \ref{Prop:Representation}, each $K_i^{\bm{\pm}\mathbf{d}}$ is represented by $E_{P_i}$, i.e., the symmetric part of the left trace of $P_i$. By part (i) of the same proposition, the operator $(\cup_i K_i^{\bm{\pm}\mathbf{d}})^{\mathbf{d}}$ is represented by $\cap_i E_{P_i}$.
\end{itemize}
\qed

\subsection{Ordering revision types}\label{Proof:K1}
The lemma below forms the basis of the Hasse diagram in Figure \ref{fig:HassePreorder}.

\begin{lemma}\label{lemma:K1}
Let $K_1, \dots, K_n$ be $n\geq 1$ knowledge operators in  $\mathcal{K}^1$, and let $\mathbf{P} = \left(P_{1}, \dots, P_{n}\right)$ be the profile of relations in $\mathcal{P}^1$ that represent those $n$ knowledge operators. Then the following are true:
\begin{itemize}
\item[(i)] $(\cup_{i} K_i)^{\mathbf{d}} \leq (\cup_{i} K_i)^{\mathbf{d}\bm{+}\mathbf{d}} \leq (\cup_{i} K_i)^{\mathbf{d}\bm{\pm}\mathbf{d}} \leq (\cup_{i} K_i)^{\bm{\pm}\mathbf{d}} \leq (\cup_{i} K_i^{\bm{\pm}\mathbf{d}})^{\mathbf{d}}$;
\item[(ii)] $(\cup_{i} K_i)^{\mathbf{d}\bm{+}\mathbf{d}} \leq (\cup_{i} K_i)^{\bm{+}\mathbf{d}} \leq (\cup_{i} K_i)^{\bm{\pm}\mathbf{d}}$;
\item[(iii)] $(\cup_{i} K_i)^{\bm{+}\mathbf{d}} \leq (\cup_{i} K_i^{\bm{+}\mathbf{d}})^{\mathbf{d}} \leq (\cup_{i} K_i^{\bm{\pm}\mathbf{d}})^{\mathbf{d}}$.
\end{itemize}
Dually, also the following are true:
\begin{itemize}
\item[(i')] $\cap_{i} P_{i} \supseteq T_{\cap_{i} P_{i}} \supseteq E_{\cap_{i} P_{i}} \supseteq E_{\mathbf{P}} \supseteq \cap_i E_{P_i}$;
\item[(ii')] $T_{\cap_{i} P_{i}} \supseteq T_{\mathbf{P}} \supseteq E_{\mathbf{P}}$;
\item[(iii')] $T_{\mathbf{P}} \supseteq \cap_i T_{P_i} \supseteq \cap_i E_{P_i}$.
\end{itemize}
\end{lemma}

\begin{proof}
It is enough to show the dual statements (i'), (ii') and (iii'). It is immediate that $\cap_{i} P_{i} \supseteq T_{\cap_{i} P_{i}} \supseteq E_{\cap_{i} P_{i}}$. It is also immediate that $T_{\mathbf{P}} \supseteq E_{\mathbf{P}}$ and $\cap_i T_{P_i} \supseteq \cap_i E_{P_i}$. To show $T_{\cap_{i} P_{i}} \supseteq T_{\mathbf{P}}$, notice that $(\omega, \omega')\in T_{\mathbf{P}}$ implies $\cap_j P_{j}(\omega') \subseteq P_{i}(\omega)$ for all $i \in \{1,\dots, n\}$. Therefore, $\cap_j P_{j}(\omega') \subseteq \cap_i P_{i}(\omega)$ and $(\omega, \omega') \in T_{\cap_{i} P_{i}}$. Furthermore, $E_{\cap_{i} P_{i}} \supseteq E_{\mathbf{P}}$ follows easily from $T_{\cap_{i} P_{i}} \supseteq T_{\mathbf{P}}$. To show $T_{\mathbf{P}} \supseteq \cap_i T_{P_i}$, notice that if $(\omega, \omega')\in \cap_i T_{P_i}$, then $P_i(\omega') \subseteq P_i(\omega)$ for all $i\in \{1,\dots, n\}$, from which it follows that $(\omega, \omega') \in T_{\mathbf{P}}$. It remains to show that $E_{\mathbf{P}} \supseteq \cap_i E_{P_i}$. If $(\omega, \omega') \in \cap_i E_{P_i}$, then $P_i(\omega') = P_i(\omega)$ for all $i \in \{1,\dots, n\}$. As a consequence, $\mathsf{Min} (\omega) = \mathsf{Min}(\omega')$ and $(\omega, \omega') \in E_{\mathbf{P}}$.
\end{proof}

\subsection{Proof of the Theorem in Section \ref{Sec:DK}}\label{Proof:THM}
There are two cases to consider. In the first, $\bar{\theta} = \theta_i$ for some $i\in I$. Without loss of generality, suppose $\bar{\theta} = \theta_N$. Hence, for all $i\in I$,
\begin{equation*}
K_i^1 = (K_1, \dots, K_N)^{\theta_i} \leq (K_1, \dots, K_N)^{\bar{\theta}} = K_N^1.
\end{equation*}
Notice that $\cup_{i\in I} K_i^1 = K_N^1$. In addition, we have $\cup_{i\in I} (K_i^1)^{\bm{+} \mathbf{d}} = K_N^1$ when $\theta_N = (\bm{+} \mathbf{d}\text{s})$ and $\cup_{i\in I} (K_i^1)^{\bm{\pm} \mathbf{d}} = K_N^1$ when $\theta_N = (\bm{\pm} \mathbf{d}\text{s})$. Then it is easy to check that for all $i\in I$,
\begin{equation*}
K_N^1 \leq K_i^2 = (K_1^1, \dots, K_N^1)^{\theta_i} \leq  (K_1^1, \dots, K_N^1)^{\bar{\theta}} = K_N^2 = K_N^1.
\end{equation*}

This means that everyone has the same knowledge $K_N^1$ at the end of the second stage of communication. Consequently, no one will strictly expand her knowledge in subsequent communication stages and $K_i^t = K_N^1$ for all $i\in I$ and $t \geq 2$. Therefore, we have $K_{\mathbf{D}} = K_N^1 = (K_1, \dots, K_N)^{\bar{\theta}}$.

In the second case, we have $\bar{\theta} \neq \theta_i$ for all $i\in I$. This means that there are exactly two maximal types in the group. Without loss of generality, suppose the two maximal types are $\theta_{N - 1}$ and $\theta_{N}$. Three subcases are possible. First, suppose $\theta_{N - 1} = (\bm{+} \mathbf{d}\text{s})$, $\theta_{N} = (\mathbf{d}\bm{\pm}\mathbf{d})$, and $\theta_{i} \neq (\bm{\pm} \mathbf{d}\text{s})$, $(\bm{\pm} \mathbf{d})$ for all $i\in I$. Notice that $\bar{\theta} = (\bm{\pm} \mathbf{d}\text{s})$. By Lemma \ref{lemma:K1}, for all $i \in I$, either
\begin{equation*}
	K_i^1 = (K_1, \dots, K_N)^{\theta_i} \leq (K_1, \dots, K_N)^{(\bm{+} \mathbf{d}\text{s})} = K_{N-1}^1
\end{equation*}
or
\begin{equation*}
	K_i^1 = (K_1, \dots, K_N)^{\theta_i} \leq (K_1, \dots, K_N)^{(\mathbf{d}\bm{\pm}\mathbf{d})} = K_{N}^1
\end{equation*}
or both. Then
\begin{equation*}
	K_{N-1}^2 = (K_1^1, \dots, K_N^1)^{(\bm{+} \mathbf{d}\text{s})} = (\cup_{i\in I} (K_i^1)^{\bm{+}\mathbf{d}})^{\mathbf{d}} = (K_{N-1}^1 \cup K_{N}^1)^{\mathbf{d}}
\end{equation*}
and
\begin{equation*}
K_{N}^2 = (K_1^1, \dots, K_N^1)^{(\mathbf{d}\bm{\pm}\mathbf{d})} = (\cup_{i\in I} K_i^1)^{\mathbf{d}\bm{\pm}\mathbf{d}} = (K_{N-1}^1 \cup K_{N}^1)^{\mathbf{d}\bm{\pm}\mathbf{d}}.
\end{equation*}
It follows from Lemma \ref{lemma:K1} that $K_{N-1}^2\leq K_{N}^2$. Thus for all $i\in I$,
\begin{equation*}
	K_N^2 \leq K_i^3 = (K_1^2, \dots, K_N^2)^{\theta_i} \leq  (K_1^2, \dots, K_N^2)^{\theta_N} = K_N^3 = K_N^2,
\end{equation*}
from which it follows that $K_i^t = K_N^2$ for all $i\in I$ and $t \geq 3$, and that $K_{\mathbf{D}} = K_N^2$. It remains to show that $K_N^2 = (K_1, \dots, K_N)^{\bar{\theta}}$. Recall that $K_N^2 = (K_{N-1}^1 \cup K_{N}^1)^{\mathbf{d}\bm{\pm}\mathbf{d}}$. Let $\mathbf{P} = (P_1, \dots, P_N)$ be the profile of possibility relations representing the initial knowledge operators. By Proposition \ref{Prop:Representation}, the operator $K_{N - 1}^1$ is represented by $\cap_{i} T_{P_i}$, the operator $K_{N}^1$ is represented by $E_{\cap_i P_i}$, and $K_N^2$ is represented by the symmetric part of the left trace of $(\cap_{i} T_{P_i}) \cap (E_{\cap_i P_i})$. Now, the relation $(\cap_{i} T_{P_i}) \cap (E_{\cap_i P_i})$ is reflexive and transitive because it is the intersection of two reflexive and transitive relations. Consequently, the left trace of $(\cap_{i} T_{P_i}) \cap (E_{\cap_i P_i})$ coincides with $(\cap_{i} T_{P_i}) \cap (E_{\cap_i P_i})$, and its symmetric part is equal to $(\cap_{i} E_{P_i}) \cap (E_{\cap_i P_i})$. By Lemma \ref{lemma:K1}, $(\cap_{i} E_{P_i}) \cap (E_{\cap_i P_i}) = \cap_{i} E_{P_i}$. By Proposition \ref{Prop:Representation}, the relation $\cap_{i} E_{P_i}$ represents $(\cup_{i} K_i^{\bm{\pm}\mathbf{d}})^{\mathbf{d}} = (K_1, \dots, K_N)^{(\bm{\pm} \mathbf{d}\text{s})}$. Therefore, we can conclude that $K_{\mathbf{D}} = (K_1, \dots, K_N)^{\bar{\theta}}$.

We can follow, \textit{mutatis mutandis}, the argument above for the remaining two subcases. In the second subcase, suppose $\theta_{N - 1} = (\bm{+} \mathbf{d}\text{s})$, $\theta_{N} = (\bm{\pm}\mathbf{d})$, and $\theta_{i} \neq (\bm{\pm} \mathbf{d}\text{s}) = \bar{\theta}$ for all $i\in I$. Thus we have
\begin{equation}\label{eq:DKproof1}
K_{\mathbf{D}} = K_N^2 = (K_{N-1}^1 \cup K_{N}^1)^{\bm{\pm}\mathbf{d}} = \left((K_1, \dots, K_N)^{(\bm{+} \mathbf{d}\text{s})} \cup (K_1, \dots, K_N)^{(\bm{\pm}\mathbf{d})}\right)^{\bm{\pm}\mathbf{d}}.
\end{equation}
By Proposition \ref{Prop:Representation}, the operator $K_N^2$ in \eqref{eq:DKproof1} is represented by the symmetric part of the left binary trace of $(\cap_i T_{P_i}, E_{\mathbf{P}})$. Since both $\cap_i T_{P_i}$ and $E_{\mathbf{P}}$ are reflexive and transitive, the left binary trace of $(\cap_i T_{P_i}, E_{\mathbf{P}})$ is equal to $(\cap_{i} T_{P_i}) \cap (E_{\mathbf{P}})$, and its symmetric part is $(\cap_{i} E_{P_i}) \cap (E_{\mathbf{P}}) = \cap_{i} E_{P_i}$. The latter relation represents $(\cup_i K_i^{\bm{\pm} \mathbf{d}})^{\mathbf{d}} = (K_1, \dots, K_N)^{(\bm{\pm} \mathbf{d}\text{s})}$, hence $K_{\mathbf{D}} = (K_1, \dots, K_N)^{\bar{\theta}}$.

In the third and final subcase, suppose $\theta_{N - 1} = (\bm{+}\mathbf{d})$, $\theta_{N} = (\mathbf{d}\bm{\pm}\mathbf{d})$, and $\theta_i \neq (\bm{\pm}\mathbf{d})$, $(\bm{+} \mathbf{d}\text{s})$, $(\bm{\pm} \mathbf{d}\text{s})$ for all $i\in I$. Notice that $\bar{\theta} = (\bm{\pm}\mathbf{d})$. It follows from Lemma \ref{lemma:K2} in Appendix \ref{Proof:FeasRev} that $K_{N-1}^2 \leq K_{N}^2$. Thus we have
\begin{equation}\label{eq:DKproof2}
K_{\mathbf{D}} = K_N^2 = (K_{N-1}^1 \cup K_{N}^1)^{\mathbf{d}\bm{\pm}\mathbf{d}} = \left((K_1, \dots, K_N)^{(\bm{+}\mathbf{d})} \; \cup \; (K_1, \dots, K_N)^{(\mathbf{d}\bm{\pm}\mathbf{d})}\right)^{\mathbf{d}\bm{\pm}\mathbf{d}}.
\end{equation}
By Proposition \ref{Prop:Representation}, the operator $K_N^2$ in \eqref{eq:DKproof2} is represented by the symmetric part of the left trace of $(T_{\mathbf{P}}) \cap (E_{\cap_i P_i})$. Since $(T_{\mathbf{P}}) \cap (E_{\cap_i P_i})$ is the intersection of two reflexive and transitive relations, it coincides with its left trace, and its symmetric part is $(E_{\mathbf{P}}) \cap (E_{\cap_i P_i}) = E_{\mathbf{P}}$. The relation $E_{\mathbf{P}}$ represents $(\cup_i K_i)^{\bm{\pm}\mathbf{d}} = (K_1, \dots, K_N)^{(\bm{\pm}\mathbf{d})}$, from which it follows that $K_{\mathbf{D}} = (K_1, \dots, K_N)^{\bar{\theta}}$.
\qed

\newpage

\bibliographystyle{plainnat}
\bibliography{biblio}

\end{document}